\documentclass[a4paper,11pt]{article}
\usepackage{amsmath}
\usepackage{amssymb}
\usepackage{graphicx}
\usepackage{parskip}
\usepackage{amsthm}
\usepackage{fancyhdr}
\usepackage{dsfont}
\usepackage{algorithm}
\usepackage{algorithmic}
\usepackage{hyperref}
\usepackage{bm}
\usepackage{natbib}

\newtheorem{thm}{Theorem}

\newtheorem{lem}{Lemma}

\theoremstyle{definition}
\newtheorem{defn}{Definition}

\newtheorem{rmk}{Remark}
\newtheorem{ex}{Example}

\newcommand{\M}{ \mathcal{ M } }
\newcommand{\D}{ \mathcal{ D } }

\newcommand{\C}{ \mathcal{ C } }
\newcommand{\Sim}{ \mathcal{ S } }
\newcommand{\Pa}{ \mathcal{ P } }
\newcommand{\G}{ \mathcal{ G } }
\newcommand{\B}{ \mathcal{ B } }
\newcommand{\N}{ \mathbb{ N } }
\newcommand{\E}{ \mathbb{ E } }
\newcommand{\R}{ \mathbb{ R } }
\newcommand{\Prb}{ \mathbb{ P } }
\newcommand{\n}{ \mathbf{ n } }
\newcommand{\m}{ \mathbf{ m } }
\newcommand{\x}{ \mathbf{ x } }
\newcommand{\y}{ \mathbf{ y } }
\newcommand{\X}{ \mathbf{ X } }
\newcommand{\e}{ \mathbf{ e } }

\newcommand{\KL}{ \operatorname{KL} }

\newcommand{\g}{ \mathbf{ g } }
\newcommand{\s}{ \mathbf{ s } }
\newcommand{\ext}{ \operatorname{ext} }
\newcommand{\bxi}{ \bm{ \xi } }

\title{Bayesian non-parametric inference for $\Lambda$-coalescents: posterior consistency and a parametric method}

\author{Jere Koskela \\
	\texttt{koskela@math.tu-berlin.de}\\
	\small Institut f\"ur Mathematik \\
	\small Technische Universit\"at Berlin \\
	\small Stra{\ss}e des 17.~Juni 136, 10623 Berlin \\
	\small Germany
	\and
	Paul A. Jenkins \\
	\texttt{p.jenkins@warwick.ac.uk}\\
	\small Department of Statistics \\
	\small University of Warwick \\
	\small Coventry CV4 7AL \\
	\small UK
	\and
	Dario Span\`{o} \\
	\texttt{d.spano@warwick.ac.uk}\\
	\small Department of Statistics \\
	\small University of Warwick \\
	\small Coventry CV4 7AL \\
	\small UK
}
\date{\today}

\begin{document}

\maketitle

\begin{abstract}
We investigate Bayesian non-parametric inference of the $\Lambda$-measure of $\Lambda$-coalescent processes with recurrent mutation, parametrised by probability measures on the unit interval. 
We give verifiable criteria on the prior for posterior consistency when observations form a time series, and prove that any non-trivial prior is inconsistent when all observations are contemporaneous.
We then show that the likelihood  given a data set of size $n \in \N$ is constant across $\Lambda$-measures whose leading $n - 2$ moments agree, and focus on inferring truncated sequences of moments.
We provide a large class of functionals which can be extremised using finite computation given a credible region of posterior truncated moment sequences, and a pseudo-marginal Metropolis-Hastings algorithm for sampling the posterior.
Finally, we compare the efficiency of the exact and noisy pseudo-marginal algorithms with and without delayed acceptance acceleration using a simulation study.
\end{abstract}

\section{Introduction}\label{introduction}

The $\Lambda$-coalescent family is a class of coalescent processes parametrised by probability measures on the unit interval, $\Lambda \in \M_1( [ 0, 1 ] )$, introduced by \citet{Donnelly99}, \citet{Pitman99} and \citet{Sagitov99}.
We focus in particular on the family of $\Lambda$-coalescents with recurrent, finite sites, finite alleles mutation, which we refer to simply as ``$\Lambda$-coalescents" throughout this paper.
This class of processes will be introduced formally in Section \ref{preliminaries}.

In recent years, $\Lambda$-coalescents have gained prominence as population genetic models for species with a highly skewed family size distribution, particularly among marine species \citep{Boom94, Arnason04, Eldon06, Birkner08}, and have also been suggested as models of evolution under natural selection \citep{Neher13}.
The $\Lambda$-measure models skewness of the family size distribution.
Thus, $\Lambda$ represents an important confounding factor for inference from genetic data, as well as being a quantity of interest in its own right.
Failure to properly account for uncertainty in $\Lambda$ could lead to model misspecification and incorrect inference.
Consequently, likelihood-based inference for $\Lambda$-coalescents has also been an active area of research \citep{Birkner08, Birkner11, Koskela15}.
A review of $\Lambda$-coalescents and their use in population genetic inference can be found in \citep{Birkner09b}, and \citet{Steinrucken13a} present a review of $\operatorname{Beta}$-coalescent models for marine species.
In this paper we make the following contributions:
\begin{enumerate}
\item We provide the first non-parametric analysis of inferring $\Lambda$-measures from genetic data. Our method is also the first instance of inferring $\Lambda$ using the Bayesian paradigm.
\item We prove inconsistency of the posterior in full generality when data is contemporaneous, and give verifiable criteria for consistency when the data set forms a time series.
\item We present an implementable parametrisation of the non-parametric inference problem by quotienting the infinite dimensional space $\M_1( [ 0, 1 ] )$ in a suitable, data-driven way. We believe this quotienting approach to have utility in infinite dimensional inference beyond the context of this work.
\item We implement a pseudo-marginal MCMC algorithm for sampling the posterior, and provide an illustrative simulation study which demonstrates the feasibility of the algorithm for inference.
\end{enumerate}

The usual approach to inference is to focus on a parametric family of $\Lambda$-measures and infer the parameters from observations.
Common choices of parametric family are $\Lambda( dr ) = 2 ( 2 + \psi^2 )^{-1} \delta_0( dr ) + \psi^2 ( 2 + \psi^2 )^{-1} \delta_{ \psi }( dr )$ where $\psi \in (0, 1]$ \citep{Eldon06}, $\Lambda = \operatorname{Beta}(2-\alpha, \alpha)$ where $\alpha \in (1,2)$ \citep{Birkner08, Birkner11} and $\Lambda( dr ) = c \delta_0( dr ) + 2^{-1} ( 1 - c ) r dr$ where $c \in [ 0, 1 ]$ \citep{Durrett05}.
We adopt the Bayesian non-parametric approach to circumvent restrictive, finite-dimensional parametrisations.
We show in Theorem \ref{consistency} that the posterior is inconsistent in the typical setting of sampling from a stationary population at a fixed time under any non-trivial prior (i.e.~one which does not assign full mass to the truth), including parametric families.
In Theorem \ref{ts_consistency} we adapt a result of \citep{Koskela17} to provide verifiable criteria on the prior for posterior consistency when time series data is available.
We also show in Section \ref{priors} that the popular Dirichlet process mixture model prior \citep{Lo84} satisfies these conditions.
Recent advances in sequencing technology have made genetic time series data available \citep{Drummond02, Beaumont03, Anderson05, Drummond05, Bollback08, Minin08, Malaspinas12, Mathieson13}, and our results provide strong motivation for its continued use and development.

In Section \ref{parameters} we make use of the fact that a sample of size $n \in \N$ carries information about $\Lambda \in \M_1( [ 0, 1 ] )$ only via its first $n - 2$ moments to make progress towards an implementable, non-parametric algorithm.
This fact is implicit in well known sampling recursions for $\Lambda$-coalescent processes \citep{Mohle06, Birkner08, Birkner09b} and made explicit in Lemma \ref{unidentifiability}.
It has two important consequences.
Firstly, it is natural to parametrise the inference problem with truncated moment sequences because any variation in the posterior between two $\Lambda$-measures whose first $n - 2$ moments coincide is due solely to the prior.
Thus we obtain parametric inference algorithms requiring no discretisation or truncation, which are nevertheless as general as any non-parametric method.
Secondly, while the conditions imposed on the prior by Theorem \ref{ts_consistency} are restrictive when viewed in $\M_1( [ 0, 1 ] )$ (e.g.~they rule out all of the parametric families listed above), they are sufficiently mild that priors whose support contains an arbitrarily good approximation of the truncated moment sequence of any desired $\Lambda \in \M_1( [ 0, 1 ] )$ can be constructed.
Posterior consistency of finite moment sequences is readily inherited from posterior consistency of the $\Lambda$-measure.

The truncated moment sequence approach can be thought of as regularising an underdetermined inference problem by identifying an appropriate, data-driven quotient space of parameters.
This is reminiscent of the method of \emph{likelihood-informed subspaces} \citep{Cui14}, which is a recently developed tool for efficient MCMC for inverse problems.
The difference between existing subspace approaches and our quotient space is that existing work has focused on projections onto finite dimensional subspaces which preserve ``most" of the likelihood in some sense, while our method captures the likelihood function exactly.
Hence we believe our work will have wider utility beyond the $\Lambda$-coalescent setting.

We show in Theorem \ref{moment_class_thm} that identifying $n - 2$ moments of $\Lambda$ does not correspond to identifying smaller regions in $\M_1( [ 0, 1 ] )$ with increasing $n$, when measured in either total variation or Kullback-Leibler divergence.
Hence the straightforward approach of computing a maximum a posteriori moment sequence, identifying a candidate $\Lambda^* \in \M_1( [ 0, 1 ] )$ with that moment sequence and treating $\Lambda^*$ as a point estimator is inappropriate.
Instead we use results by \citet{Winkler88} to provide a broad class of functionals of $\Lambda$ which can be maximised or minimised over credible regions of the posterior.
As a simple example, it is often of great interest whether the Kingman coalescent \citep{Kingman82}, $\Lambda = \delta_0$, provides an adequate model for observed genetic data.
Any parametric family containing $\delta_0$ as a special case could be used to assess the Kingman assumption, but such an approach requires justification of the parametric family.
It is also possible that $\delta_0$ is a good model within the family, but a poor fit in some broader one.
Our method avoids both of these problems.

Finally, we provide a pseudo-marginal Metropolis-Hastings algorithm \citep{Beaumont03, Andrieu09} for sampling truncated moment sequences from the posterior distribution.
We compare the performance of the standard pseudo-marginal algorithm, the noisy pseudo-marginal algorithm, as well as delayed acceptance accelerated versions of both algorithms \citep{Christen05} in order to reduce the number of expensive likelihood estimations and improve the computational feasibility of our inference.
We find that the noisy algorithm does not improve upon standard pseudo-marginal inference in this setting, especially given that twice as many likelihood evaluations are required.
In contrast, an off-the-shelf delayed acceptance step dramatically speeds up computations.
These results are illustrated with a simulation study comparing the Kingman hypothesis, $\Lambda = \delta_0$, based on simulated data from both the Kingman and Bolthausen-Sznitman ($\Lambda = U(0, 1)$) coalescents.
The Kingman coalescent is the classical null model of neutral evolution, while the Bolthausen-Sznitman coalescent has recently emerged as an alternative model in the presence of selection \citep{Schweinsberg15} or population expansion into uninhabited territory \citep{Berestycki13b}.
In particular, the Bolthausen-Sznitman coalescent has been suggested as a model for the genetic ancestry of microbial populations such as influenza or HIV \citep{Neher13}.

The rest of the paper is laid out as follows.
Section \ref{preliminaries} provides an introduction to $\Lambda$-coalescents, $\Lambda$-Fleming-Viot jump-diffusions and a duality relation connecting the two.
In Section \ref{inconsistency} we state and prove our consistency results.
Section \ref{parameters} presents our parametrisation via moments and shows that it preserves all information in the data.
Section \ref{priors} contains example families of priors which satisfy both consistency and tractable push-forward priors on moment sequences.
Section \ref{positive_results} contains our results on inferring functionals of $\Lambda$ based on finitely many moments.
In Section \ref{simulation} we present the pseudo-marginal Metropolis-Hastings algorithm for sampling posterior distributions of moment sequences, and an accompanying simulation study which demonstrates the practicality of the method.
Section \ref{discussion} concludes with a discussion.

\section{Preliminaries}\label{preliminaries}

Let $[ d ] := \{ 1, \ldots, d \}$ represent $d$ genetic types (e.g.~$d = 4^l$ for $l$ loci of DNA), $M = ( M_{ i j } )_{ i, j = 1 }^d$ be a stochastic matrix specifying mutation probabilities between types, $\theta > 0$ be the mutation rate and $\Lambda \in \M_1( [ 0, 1 ] )$ denote the $\Lambda$-measure.
Here and throughout we assume $M$ has a unique stationary distribution $m \in \M_1( [ d ] )$.

Let $\X := \{ \X_t \}_{ t \geq 0 }$ denote the $\Lambda$-Fleming-Viot process with recurrent, finite sites, finite alleles mutation: a jump-diffusion on  the $d$-dimensional probability simplex $\Sim_d := \{ \x \in [ 0, 1 ]^d : \sum_{ i = 1 }^d x_i = 1 \}$ with generator
\begin{align}
\G^{ \Lambda } f( \x ) = & \frac{ \Lambda( \{ 0 \} ) }{ 2 } \sum_{ i, j = 1 }^d x_i ( \delta_{ i j } - x_j ) \frac{ \partial^2 }{ \partial x_i \partial x_j } f( \x ) +  \theta \sum_{ i, j = 1 }^d x_j ( M_{ j i } - \delta_{ i j } ) \frac{ \partial }{ \partial x_i } f( \x ) \nonumber\\
& + \sum_{ i = 1 }^d x_i \int_{ ( 0, 1 ] } \left[ f( ( 1 - r ) \x + r \e_i ) - f( \x ) \right] r^{ -2 } \Lambda( dr ) \label{lfv_gen}
\end{align}
acting on functions $f \in C^2( \Sim_d )$.
This process is a model for the distribution of genetic types $[ d ]$ in a large population undergoing recurrent mutation and random mating with high-fecundity reproduction events, in which a single individual becomes ancestral to a non-trivial fraction $r \in (0, 1]$ of the whole population.
As with $\Lambda$-coalescents, we will abbreviate ``$\Lambda$-Fleming-Viot process with recurrent, finite sites, finite alleles mutation" to just ``$\Lambda$-Fleming-Viot process" throughout this paper.

The first term on the right hand side (henceforth R.H.S.) of \eqref{lfv_gen} models diffusion of the allele frequencies due to random mating, or genetic drift in the terminology of population genetics.
The second term models mutation, and the third (jump) term models high-fecundity reproductive events.
Without the jump term, i.e.~when $\Lambda = \delta_0$, $\{ \X_t \}_{ t \geq 0 }$ reduces to the classical $d$-dimensional Wright-Fisher diffusion with recurrent mutation.
We denote the law of a $\Lambda$-Fleming-Viot process with initial condition $\x \in \Sim_d$ by $\Prb_{ \x }^{ \Lambda }( \cdot )$, and expectation with respect to this law by $\E_{ \x }^{ \Lambda }\left[ \cdot \right]$.
We suppress dependence on initial conditions whenever the stationary process is meant.
For bounded $f: \Sim_d \mapsto \R$, let $P_t^{ \Lambda } f( \x ) := \E_{ \x }^{ \Lambda }[ f( \X_t ) ]$ be the associated transition semigroup, $p_t^{ \Lambda }( \x, \y )$ be the transition density and $\pi^{ \Lambda }( \x )$ be the corresponding stationary density on $\Sim_d$, assumed unique.
The transition semigroup is Feller for any $\Lambda \in \M_1( [ 0, 1 ] )$ \citep{Bertoin03}, and all densities are assumed to exist with respect to the $d$-dimensional Lebesgue measure $d\x$.

A realisation of a $\Lambda$-Fleming-Viot process specifies the relative frequencies of genetic types in an  infinite population across time.
It is natural to imagine sampling $n \in \N$ individuals from the population at a fixed time, and ask about the ancestral tree connecting the sampled individuals.
This ancestral tree is a random object, and is described by the $\Lambda$-coalescent, denoted by $\Pi := \{ \Pi_t \}_{ t \geq 0 }$, taking values in $\Pa_n^d$, the set of $[ d ]$-labelled partitions of $[ n ]$ \citep[Section 5]{Donnelly99}.
The process is started from the unlabelled partition $\psi_n := \{ \{ 1 \}, \{ 2\}, \ldots, \{ n \} \}$, and propagates backwards in time from the point of view of the reproductive evolution of the population.
When the process has $p \in \N$ blocks, any $2 \leq k \leq p$ of them merge at rate
\begin{equation*}
\lambda_{p, k} := \Lambda( \{ 0 \} ) \mathds{ 1 }_{ \{ 2 \} }( k ) + \int_{ ( 0, 1] } ( 1 - r )^{ p - k } r^{ k - 2 } \Lambda( dr ),
\end{equation*}
where $\mathds{ 1 }_A( k )$ is the indicator function of the set $A$.
We will abbreviate $\lambda_{ k, k } \equiv \lambda_k$ throughout the paper.
Once the process has merged into a single block, known as the most recent common ancestor (MRCA), an ancestral type is sampled from an initial law $\nu \in \M_1( [ d ] )$.
We focus on the case of a stationary population, in which case $\nu = m$.
This type is inherited along the branches of the $\Lambda$-coalescent tree, with mutations occurring at rate $\theta$ and mutant types sampled from $M$.
The result is a random labelling of partitions for all times $0 \leq t \leq T$, where $T$ denotes the hitting time of the MRCA.
Let $\mathbf{ P }_{ n }^{ \Lambda }$ denote the law of $\Pi$ started from $\psi_n$, and $E_{ n }^{ \Lambda }$ denote the corresponding expectation.
For the entirety of the paper we assume $\theta$ and $M$ are known, and focus on inferring $\Lambda$.
This assumption is crucial for the proof of consistency of Bayesian inference from time series data (Theorem \ref{ts_consistency} in Section \ref{inconsistency}), but not needed for correctness of the finitely many moments parametrisation in Section \ref{parameters}, or of the algorithms in Section \ref{simulation}.

The following relationship between $\Lambda$-Fleming-Viot processes and corresponding $\Lambda$-coalescents is classical \citep{Bertoin03}, and will be useful in the (in)consistency proofs in the next section:
\begin{equation}\label{duality_eq}
\E^{ \Lambda } \left[ \prod_{ i = 1 }^d X_t( i )^{ n_i } \right] = E_n^{ \Lambda }\left[ \prod_{ i = 1 }^d m( i )^{ | \Pi_t( i ) | } \right],
\end{equation}
where $n_i$ denotes the number of observed individuals with label $i \in [d]$ sampled i.i.d.~from the random measure $\X_t$, and $| \Pi_t( i ) |$ denotes the number of blocks in partition $\Pi_t$ with label $i \in [d]$.
Formula \eqref{duality_eq} is an example of so called \emph{moment duality} between stochastic processes (see e.g.~\citep{Mohle99}, and references therein for details):
Intuitively, \eqref{duality_eq} states that the law of the type frequencies of $\Lambda$-coalescent leaves coincides with a multinomial sample from a random measure drawn from the corresponding $\Lambda$-Fleming-Viot process.

\section{Posterior consistency}\label{inconsistency}

Let $Q \in \M_1( \M_1( [ 0, 1 ] ) )$ be a prior distribution for $\Lambda$, and $\n = ( n_1, \ldots, n_d ) \in \N^d$ denote observed type frequencies of $n := \sum_{ i = 1 }^d n_i$ $[ d ]$-labelled lineages generated by the $\Lambda$-coalescent.
For Borel sets $A \in \mathcal{B}( \M_1( [ 0, 1 ] ) )$, define the posterior as
\begin{equation*}
Q( A | \n ) = \frac{ \int_A \mathbf{ P }_n^{ \Lambda }( \n ) Q( d\Lambda ) }{ \int_{ \M_1( [ 0, 1 ] ) } \mathbf{ P }_n^{ \Lambda }( \n ) Q( d\Lambda ) }.
\end{equation*}
Informally, posterior consistency holds when $Q( \cdot | \n )$ concentrates on a neighbourhood of the $\Lambda_0 \in \M_1( [ 0, 1 ] )$ which generated $\n$ as $n \rightarrow \infty$.
This is a natural requirement for statistical inference as it ensures the truth can be learned from a sufficient amount of data.
For an overview of Bayesian non-parametric statistics, the reader is directed at \citep{Hjort10} and references therein.

It is well known that non-parametric posterior consistency is highly sensitive to the details of the topology defining the neighbourhood system as well as the mode of convergence \citep{Diaconis86}.
This will also be the case for our consistency result for time series data, Theorem \ref{ts_consistency}.
In contrast, the inconsistency result for contemporaneous observations, Theorem \ref{consistency}, is very universal.
Hence we postpone specification of these details until after Theorem \ref{consistency}.
\vskip 11pt
\begin{rmk}
In Theorem \ref{consistency} below, we give a formal statement of inconsistency from the point of view of Bayesian estimators, which are our interest in this paper.
Elements of the proof of Theorem \ref{consistency} will also be useful in the proving Theorem \ref{ts_consistency} later in the paper.
However, we emphasize that the negative result also holds for frequentist estimators based on contemporaneous data.
Essentially the same argument used to prove Theorem \ref{consistency} shows that the limiting likelihood
\begin{equation*}
\lim_{ n \rightarrow \infty } \E^{ \Lambda }[ q( \n | \X ) ] = \pi^{ \Lambda }( \x )
\end{equation*}
is positive for any observation $\x$ and any $\Lambda$, at least provided $\{ \pi^{\Lambda} \}_{ \Lambda }$ is a bounded family of stationary densities so that the Dominated Convergence Theorem holds without the regularising effect of the prior.
Hence, estimators cannot converge to the true $\Lambda$-measure generating the data.
Ultimately, the problem is that the $d$-dimensional vector corresponding to one exact draw from a stationary $\Lambda$-Fleming-Viot process cannot be expected to uniquely identify an infinite dimensional object.
Similar problems of identifiability also occur in estimation of lifetime distributions in branching processes \citep{Hopfner02, Hoffmann16}, for which identifiability issues can be overcome by letting the length of the observation window grow to infinity, analogously to Theorem \ref{ts_consistency} below.
\end{rmk}
\vskip 11pt
\begin{thm}\label{consistency}
Let $\n \in \N^d$ denote the observed type frequencies in a sample of size $n \in \N$ generated by a $\Lambda$-coalescent started from $\psi_n$ at a fixed time, and let $\x := \lim_{ n \rightarrow \infty } \frac{ \n }{ n } \in \M_1( [ d ] )$ denote the limiting observed relative type frequencies.
Then the limiting posterior is given by
\begin{equation*}
\lim_{ n \rightarrow \infty } Q( A | \n ) = \frac{ \int_A \pi^{ \Lambda }( \x ) Q( d\Lambda ) }{ \int_{ \M_1( [ 0, 1 ] ) } \pi^{ \Lambda }( \x ) Q( d\Lambda ) }.
\end{equation*}
In particular, the R.H.S.~is positive for any $A \in \mathcal{ B }( \M_1( [ 0, 1 ] ) )$ which intersects the support of $Q$, regardless of the $\Lambda \in \M_1( [ 0, 1 ] )$ generating the data.
\end{thm}
\begin{proof}
Conditioning on the ancestral tree of the observed sample give the following representation for the posterior:
\begin{equation*}
Q( A | \n ) = \frac{ \int_A \mathbf{ P }_{ n }^{ \Lambda }( \n ) Q( d\Lambda ) }{ \int_{ \M_1( [ 0, 1 ] ) } \mathbf{ P }_{ n }^{ \Lambda }( \n ) Q( d\Lambda ) } = \frac{ \int_A E_{ n }^{ \Lambda }\left[ \mathds{ 1 }_{ \{ \n \} }( \Pi_0 ) \right] Q( d\Lambda ) }{ \int_{ \M_1( [ 0, 1 ] ) } E_{ n }^{ \Lambda }\left[ \mathds{ 1 }_{ \{ \n \} }( \Pi_0 ) \right] Q( d\Lambda ) }.
\end{equation*}
Using \eqref{duality_eq} we can write
\begin{equation*}
Q( A | \n ) = \frac{ \int_A \E^{ \Lambda }\left[  q( \n | \X ) \right] Q( d\Lambda ) }{ \int_{ \M_1( [ 0, 1 ] ) } \E^{ \Lambda }\left[  q( \n | \X ) \right] Q( d\Lambda ) },
\end{equation*}
where $q( \n | \X ) := \binom{ n }{ n_1 , \ldots, n_d } \prod_{ i = 1 }^d X_i^{ n_i }$ is the multinomial sampling probability.
We will show the requisite convergence of the numerator and denominator separately, and the result will follow by the algebra of limits.
We begin by considering the numerator.

By Fubini's theorem
\begin{equation*}
\int_A \E^{ \Lambda }\left[  q( \n | \X ) \right] Q( d\Lambda ) = \int_{ \Sim_d } q( \n | \y ) \int_A \pi^{ \Lambda }( \y ) Q( d\Lambda ) d\y =: \int_{ \Sim_d } q( \n | \y ) F_{ Q; A }( \y ) d\y,
\end{equation*}
where $F_{ Q; A }( \y )$ is a sub-probability density on $\Sim_d$ since it is a mixture of probability densities.
Hence $F_{ Q; A }( \y ) d\y$ defines a finite measure on $\Sim_d$, and $q( \n | \y ) \leq 1$ so that by the Dominated Convergence theorem
\begin{equation*}
\lim_{ n \rightarrow \infty } \int_{ \Sim_d } q( \n | \y ) F_{ Q; A }( \y ) d\y  = \int_{ \Sim_d } \lim_{ n \rightarrow \infty } q( \n | \y ) F_{ Q; A }( \y ) d\y.
\end{equation*}

By the Law of Large Numbers $\n \sim \lfloor n \x \rfloor$ so that $q( \n | \y ) \rightarrow q( \lfloor n \x \rfloor | \y )$, and by Stirling's formula
\begin{equation*}
q( \lfloor n \x \rfloor | \y ) \sim \prod_{ i = 1 }^d \left( \frac{ y_i }{ x_i } \right)^{ n x_i },
\end{equation*}
or
\begin{equation*}
\log( q( \lfloor n \x \rfloor | \y ) ) \sim n \sum_{ i = 1 }^d x_i \log\left( \frac{ y_i }{ x_i } \right) = - n \sum_{ i = 1 }^d x_i \log\left( \frac{ x_i }{ y_i } \right) = - n \operatorname{KL}( \x, \y ),
\end{equation*}
where $\operatorname{KL}( \x, \y )$ denotes the Kullback-Leibler divergence between the probability mass functions $\x$ and $\y$.
By Gibbs' inequality $\operatorname{KL}( \x, \y ) \geq 0$ and $\operatorname{KL}( \x, \y ) = 0$ if and only if $\x = \y$, so that $q( \n | \cdot ) \rightarrow \delta_{ \x }( \cdot )$ almost surely.
Hence 
\begin{equation*}
\int_{ \Sim_d } \lim_{ n \rightarrow \infty } q( \n | \y ) F_{ Q; A }( \y ) d\y = F_{ Q; A }( \x ) = \int_A \pi^{ \Lambda }( \x ) Q( d\Lambda ),
\end{equation*}
as required.
The argument for the denominator is identical after substituting $\M_1( [ 0, 1 ] )$ for the domain of integration $A$.
\end{proof}
\begin{rmk}\label{yuns_result}
There is an apparent contradiction between the negative conclusion of Theorem \ref{consistency} and recent positive results \citep[Theorems 2, 3, 4 and 5]{Spence16} showing that $\Lambda$-measures can often be identified from their site frequency spectra.
The contradiction is resolved by noting that \citet{Spence16} work directly with the expected site frequency spectrum, thereby sidestepping both the randomness of the ancestral tree and the randomness of the mutation process given the tree.
Numerical investigations by \citet{Spence16} show that their method is unreliable unless a modest number (10-100) of independent realisations of ancestral trees is available.
Independent trees cannot be sampled from populations whose ancestry is described by any non-Kingman $\Lambda$-coalescent, even in the idealised scenario of an infinitely long genome in the presence of recombination.
However, as noted by \citep{Spence16}, in practice the decay of correlations with increasing genome length is determined by the prelimiting model of evolution, and not necessarily the limiting $\Lambda$-coalescent.
For example, the selective sweep model of \citet{Durrett05} can allow for asymptotically independent trees across a genome in the presence of multiple mergers for some combinations of parameters, in which case the identifiability results of \citet{Spence16} hold.
\end{rmk}

The following example is an extension of a result by \citet{Der14}, and demonstrates that the lack of consistency can have dramatic consequences for statistical identifiability even in the case of very simple priors.
\vskip 11pt
\begin{ex}\label{ex_inconsistency}
Consider $d = 2$, $M_{ i j } = 1/2$ for $i, j \in \{ 1, 2 \}$, and set $Q( d\Lambda ) = \frac{ 1 }{ 2 } \delta_{ \delta_0 }( d\Lambda ) + \frac{ 1 }{ 2 } \delta_{ \delta_1 }( d\Lambda )$.
The stationary law $\pi^{ \Lambda }( x )$ is known in the parent-independent, two-allele case for both of these atoms \citep{Der14}:
\begin{align*}
\pi^{ \delta_0 }( x ) &= \frac{ \Gamma( 2 \theta ) }{ \left[ \Gamma( \theta ) \right]^2 } x^{ \theta - 1 } ( 1 - x )^{ \theta - 1 } \\
\pi^{ \delta_1 }( x ) &= \frac{ 1 }{ \theta } | 1 - 2 x |^{ \frac{ 1 - \theta }{ \theta } },
\end{align*}
so the expected limiting posterior probabilities can be computed assuming either data-generating measure.
These are listed in Table \ref{ex_post} for some candidate values of $\theta$, while Figure \ref{exp_post_probs} depicts limiting posterior probabilities as functions of the observed allele frequencies.
The extreme sensitivity of the posterior probabilities in Figure \ref{exp_post_probs} is akin to the ``Bayesian brittleness" investigated by \citet{Owhadi15}, resulting in inferences which are not robust to small changes in the observed allele frequencies, prior probabilities or latent parameters.
\begin{table}[!ht]
\centering
\begin{tabular}{c|c|c}
$\theta$ & $\E^{ \delta_0 }[ Q( \delta_0 | X ) ]$ & $\E^{ \delta_1 }[ Q( \delta_0 | X ) ]$ \\
\hline
0.04 & 0.84 & 0.16 \\
0.1 & 0.73 & 0.27 \\
0.5 & 0.54 & 0.46 \\
1 & 0.50 & 0.50 \\
5 & 0.65 & 0.35 \\
10 & 0.75 & 0.25 \\
17 & 0.82 &0.18
\end{tabular}
\caption{Expected posterior probabilities given an infinite number of contemporaneous observations in the parent-independent, two-allele model.}
\label{ex_post}
\end{table}
\begin{figure}[!ht]
\centering
\includegraphics[width = 0.49 \linewidth]{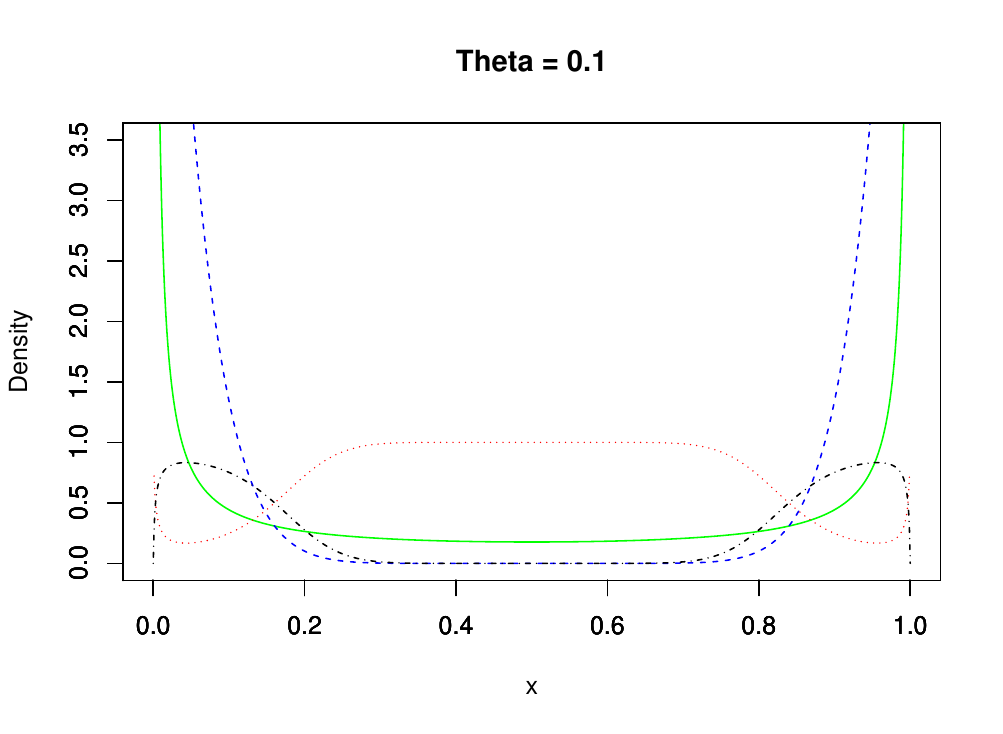}
\includegraphics[width = 0.49 \linewidth]{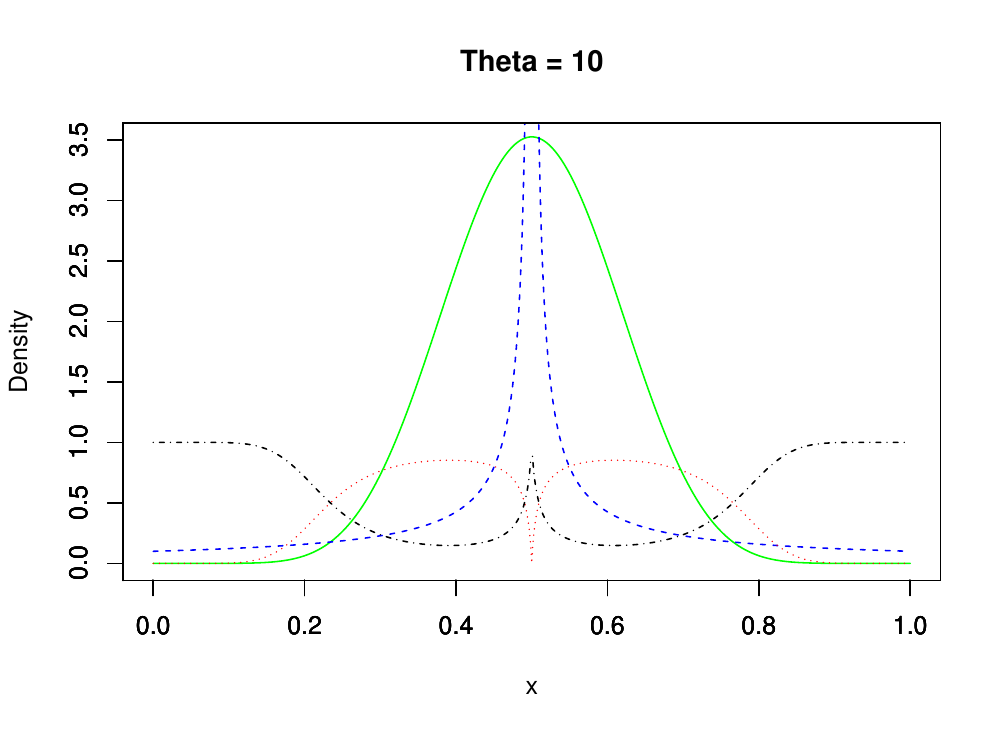}
\caption{Limiting posterior probabilities as functions of the observed allele frequency $x$ for $\theta = 0.1$ and $10$ in the parent-independent, two-allele model. $Q( \delta_0 | x )$ is dotted, $Q( \delta_1 | x )$ is dash-dotted, $\pi^{ \delta_0 }( x )$ is solid and $\pi^{ \delta_1 }( x )$ is dashed. Note the extreme sensitivity of the posterior to the observed allele frequencies near $x = 0.5$ when $\theta = 10$ and near $x = 0$ or $1$ when $\theta = 0.1$.}
\label{exp_post_probs}
\end{figure}
\end{ex}

The fact that $\pi^{ \delta_0 }( x ) = \pi^{ \delta_1 }( x )$ when $\theta = 1$ in Example \ref{ex_inconsistency} was pointed out by \citet{Der14} as proof of the fact that $\Lambda$-measures cannot in general be uniquely identified from independent draws from $\pi^{ \Lambda }( x )$.
Our calculations illustrate that inference suffers from low power and poor stability even when $\theta \neq 1$ if all observations are contemporaneous.

The inconsistency result of Theorem \ref{consistency} holds for essentially arbitrary priors.
Our next aim is to show that the posterior can be consistent when the data set is a time series of increasing length.
This does not contradict the unidentifiability claim of \citet{Der14}, because the authors only consider independent draws from $\pi^{ \Lambda }$.
In contrast, in our setting it is information about transition densities $p_{ \Delta }^{ \Lambda }( \x, \y )$ which facilitates posterior consistency.

We begin by defining the topology and weak posterior consistency following the set up of \citet{vanDerMeulen13}, who considered similar time series data for one dimensional diffusions.
In addition to topological details, posterior consistency is highly sensitive to the support of the prior, which should not exclude the truth.
This is usually guaranteed by insisting that the prior places positive mass on all neighbourhoods of the truth, typically measured in terms of Kullback-Leibler divergence.
In our setting such a support condition is provided by \eqref{kl_cond} below.
\vskip 11pt
\begin{defn}\label{prior_supp}
Fix $\eta > 0$ and let $\D_{ \eta }$ be a collection of Lebesgue probability densities on $[ \eta, 1 ]$ satisfying $\inf_{ r \in [ \eta, 1 ] } \phi( r ) > 0$ and $\sup_{ r \in [ \eta, 1 ] } \phi( r ) < \infty$ for each $\phi \in \D_{ \eta }$.
We assume that $\Lambda( dr ) = \phi( r ) dr$ for $\phi \in \D_{ \eta }$, and denote the data generating density by $\phi_0$.
\end{defn}
Restricting the support of $\phi$ to $[ \eta, 1 ]$ ensures that the $\Lambda$-coalescent can have no Kingman component, and that the $\Lambda$-Fleming-Viot process is a compound Poisson process with drift.
Furthermore, most previously studied parametric families of $\Lambda$-measures are ruled out, including all those mentioned in Section \ref{introduction}.
However, we will see in Section \ref{parameters} that the prior can be chosen to satisfy the conditions of Definition \ref{prior_supp} and place mass arbitrarily close to any desired $\Lambda$-measure, or family of $\Lambda$-measures, in a way we will make precise in Example \ref{prior_example}.

Before we define and prove consistency of Bayesian nonparametric inference, we need to first establish that $\Lambda$ is identifiable from discrete observations.
This is done in Lemma \ref{phi_ident} below.
\vskip 11pt
\begin{lem}\label{phi_ident}
For any pair $\phi_0 \neq \phi \in \D_{ \eta }$ and any $\Delta > 0$ there exists $\x \in \Sim_d$ and a test function $f$ such that $P_{ \Delta }^{ \phi_0 } f( \x ) \neq P_{ \Delta }^{ \phi } f( \x )$.
In particular, identifying $P_{ \Delta }^{ \phi }$ is equivalent to identifying $\phi$.
\end{lem}
\begin{proof}
Let $\phi_0 \in \D_{ \eta }$ and $\phi \in \D_{ \eta }$ agree on their first $n - 3$ moments, and suppose that the $(n - 2)^{\text{th}}$ moments $\lambda_n^0$ and $\lambda_n$ differ, as must be the case for some $n \geq 3$ if $\phi_0 \neq \phi$.
We begin by using the spectral representation of \citet{Griffiths14} to show that certain eigenvalues of the corresponding $\Lambda$-Fleming-Viot generators $\G^{ \phi_0 }$ and $\G^{ \phi }$ are different.

By \citep[Theorem 5]{Griffiths14}, the eigenvalues are naturally indexed with multi-indices $\n \in \N^d$.
Let $q_{ \n }^0$ and $q_{ \n }$ denote the eigenvalues of $\G^{ \phi_0 }$ and $\G^{ \phi }$ with index $\n$, respectively.
Define independent random variables $( U, Y, V ) \sim 2u \mathds{ 1 }_{ [ 0, 1 ] } \otimes \phi \otimes U(0, 1)$, and let $W := U Y$.
By \citep[equation (41)]{Griffiths14}, we have that $q_{ \n }$ can be written as
\begin{equation*}
q_{ \n } = \binom{ n }{ 2 } \E[ ( 1 - W )^{ n - 2 } ] + C( \theta, M, \n ),
\end{equation*}
for each $\n$ such that $\sum_{ i = 1 }^d n_i = n$, where the expectation is with respect to the law of $W$, and $C( \theta, M, \n )$ is a constant depending only on its arguments.
A binomial expansion yields
\begin{equation*}
\E[ ( 1 - W )^{ n - 2 } ] = \sum_{ k = 2 }^{ n - 2 } \binom{ n - 2 }{ k } ( - 1 )^k \E[ W^k ] = \sum_{ k = 2 }^{ n - 2 } \binom{ n - 2 }{ k } ( - 1 )^k \frac{ 2 }{ k + 2 } \lambda_{ k + 2 }.
\end{equation*}
All terms on the R.H.S.~coincide for $q_{ \n }$ and $q_{ \n }^0$ except the $k = n - 2$ term, for which 
\begin{equation*}
\frac{ ( -1 )^{ n - 2 } 2 }{ n } \lambda_n \neq \frac{ ( -1 )^{ n - 2 } 2 }{ n } \lambda_n'
\end{equation*}
by assumption.
Hence, the eigenvalues $q_{ \m }^0$ and $q_{ \m }$ for two densities $\phi_0$ and $\phi$ coincide for multi-indices of total degree up to $m \leq n - 1$, and eigenvalues for multi-indices of total degree $n$ differ when $n - 2$ is the order of the first moment which is distinct between $\phi_0$ and $\phi$.

Next, we use the above result on eigenvalues, in conjunction with \citep[Theorem 5]{Griffiths14}, to show that $\G^{ \phi_0 }$ and $\G^{ \phi }$ also have some distinct right eigenfunctions.
To that end, let $h_{ \n }^0( \x )$ and $h_{ \n }( \x )$ be the respective eigenfunctions of $\G^{ \phi_0 }$ and $\G^{ \phi }$ with eigenvalues $q_{ \n }^0$ and $q_{ \n }$.
Applying the representation of $\G^{ \phi }$ in \citep[equation (39)]{Griffiths14} to the monomial test function
\begin{equation*}
q( \m | \x ) := \prod_{ i = 1 }^d x_i^{ m_i },
\end{equation*}
as well as a binomial expansion to the resulting $[ ( 1 - W ) x_i + W V ]^{ m_i - 1 - \delta_{ i j } }$ terms yields
\begin{align}
\G^{ \phi } &q( \m | \x ) = \sum_{ i = 1 }^d \binom{ m_i }{ 2 } \sum_{ k = 0 }^{ m_i - 2 } ( 1 - W )^{ m - k - 2 } ( W V )^k q( \m - ( k + 1 ) \e_i | \x ) \nonumber \\
&- \sum_{ i, j = 1 }^d m_i ( m_j - \delta_{ i j } ) \sum_{ k = 0 }^{ m_i - 1 - \delta_{ i j } } \binom{ m_i - 1 - \delta_{ i j } }{ k } ( 1 - W )^{ m - k - 2 } ( W V )^k q( \m - k \e_i | \x ) \nonumber \\
&+ \theta \sum_{ i = 1 }^d \left( \sum_{ j = 1 }^d M_{ j i } x_j - x_i \right) \frac{ \partial }{ \partial x_i } q( \m | \x ). \label{gen_eq}
\end{align}
The fact that the R.H.S.~depends on $\phi$ only via powers of $W$ up to order $n - 2$ makes it clear that the action of the generators $\G^{ \phi_0 }$ and $\G^{ \phi }$ coincide on polynomials of degree $m < n$, which the eigenfunctions $h_{ \m }^0( \x )$ and $h_{ \m }( \x )$ are for $m < n$ by \citep[Theorem 5]{Griffiths14}.

Now consider $h_{ \n }^0( \x )$ and $h_{ \n }( \x )$ for a fixed $\n$ of total degree $n$.
Again by \citep[Theorem 5]{Griffiths14}, the total degree $n$ terms of both equal $\bxi^{ \n }$, and hence coincide.
See \citep[Theorem 5]{Griffiths14} for the definition of $\bxi$ as a linear, $M$-dependent transformation of $\x$.
We will focus instead on terms of total degree $n - 1$, and to that effect introduce the representations
\begin{align}
h_{ \n }( \x ) &= \bxi^{ \n } + \sum_{ i = 1 }^d a_{ \n ( \n - \e_i ) } h_{ \n - \e_i }( \x ) + \sum_{ \m : m < n - 1 } a_{ \n \m } h_{ \m }( \x ) \nonumber \\
\G^{ \phi } h_{ \n }( \x )&= - q_{ \n } h_{ \n }( \x ) + \sum_{ i = 1 }^d b_{ \n ( \n - \e_i) } h_{ \n - \e_i }( \x ) + \sum_{ \m : m < n - 1 } b_{ \n \m } h_{ \m }( \x ) \label{g_h} \\
&- \sum_{ i = 1 }^d a_{ \n ( \n - \e_i ) } q_{ \n - \e_i } h_{ \n - \e_i }( \x ) - \sum_{ \m : m < n - 1 } a_{ \n \m } q_{ \m } h_{ \m }( \x ) \nonumber
\end{align}
where the coefficients $a_{ \n \m }$ and $b_{ \n \m }$ must satisfy $a_{ \n \m } q_{ \m } = b_{ \n \m }$ in order for $h_{ \n }( \x )$ to be an eigenfunction, as first three terms on the R.H.S.~of \eqref{g_h} arise as the definition of an eigenfunction expansion of $\G^{ \phi } \bxi^{ \n }$.
Let $a_{ \n \m }^0$ and $b_{ \n \m }^0$ be the analogous coefficients for the same representations of the polynomials $h_{ \n }^0$ and $\G^{ \phi_0 } h_{ \n }^0$.

For each $\m$, the only term of total degree $m$ in $h_{ \m }( \x )$ is $\bxi^{ \m }$.
Therefore, if $b_{ \n ( \n - \e_i ) }^0 \neq b_{ \n ( \n - \e_i ) }$ for some $i \in [ d ]$, then $a_{ \n ( \n - \e_i ) }^0 \neq a_{ \n ( \n - \e_i ) }$ since $q_{ \n - \e_i }^0 = q_{ \n - \e_i }$, and we are done because then $h_{ \n }( \x )$ contains at least one term which has a coefficient different to that of the corresponding term in $h_{ \n }^0( \x )$.
Terms of total degree $n - 1$ arise in the first line of \eqref{gen_eq} by taking $k = 0$, in the second by taking $k = 1$, and do not arise in the third.
In particular, the coefficient of $q( \n - \e_i | \x )$ on the R.H.S.~of \eqref{gen_eq} is
\begin{align*}
&\binom{ n_i }{ 2 } \E[ ( 1 - W )^{ n - 2 } ] - ( n - 2 ) \binom{ n_i }{ 2 } \E[ ( 1 - W )^{ n - 3 } W V ] \\
&= \binom{ n_i }{ 2 } \left[ 1 + \sum_{ k = 1 }^{ n - 2 } \binom{ n - 2 }{ k } ( -1 )^k \frac{ 2 + k }{ 2 } \lambda_{ k + 2 } \right]
\end{align*}
These coefficients all differ for $\G^{ \phi_0 }$ and $\G^{ \phi }$ because they can be written as the same linear combinations of the first $n - 2$ moments of $\phi_0$ and $\phi$, respectively.
The lower degree terms with coefficients $a_{ \n \m}$ with $m < n - 1$ cannot contribute to the coefficients of terms of degree $n - 1$.
Finally, the coefficients of $\x^{ \n - \e_i }$ and $\bxi^{ \n - \e_i }$ coincide because $\xi$ is a linear transformation of $\x$.

In short, the eigenfunction $h_{ \n }( \x )$ has the form
\begin{equation*}
h_{ \n }( \x ) = \bxi^{ \n } + \sum_{ i = 1 }^d \frac{ \binom{ n_i }{ 2 } \left[ 1 + \sum_{ k = 1 }^{ n - 2 } \binom{ n - 2 }{ k } ( -1 )^k \frac{ 2 + k }{ 2 } \lambda_{ k + 2 } \right] }{ \binom{ n - 1 }{ 2 } \E[ ( 1 - W )^{ n - 3 } ] + C_{ \n - \e_i, \theta, M } } \bxi^{ \n - \e_i } + \text{ lower order terms},
\end{equation*}
and the coefficients of the degree $n - 1$ terms all differ between $h_{ \n }^0( \x )$ and $h_{ \n }( \x )$.

Finally, fix $\Delta > 0$, and consider $P_{ \Delta }^{ \phi_0 } h_{ \n }( \x )$ and $P_{ \Delta }^{ \phi } h_{ \n }( \x ) = e^{ - q_{ \n } \Delta } h_{ \n }( \x )$.
The former is a polynomial, but cannot be a scalar multiple of the latter because otherwise $h_{ \n }( \x )$ would be an eigenfunction of $\G^{ \phi_0 }$, which we have shown is not the case.
Hence, the two polynomials $P_{ \Delta }^{ \phi_0 } h_{ \n }( \x )$ and $P_{ \Delta }^{ \phi } h_{ \n }( \x )$ are distinct, and thus differ on some non-empty, open set, which concludes the proof.
\end{proof}
\begin{defn}
Fix a sampling interval $\Delta > 0$ and a finite Borel measure $\nu$ on $\Sim_d$ placing positive mass in all non-empty open sets.
A weak topology on $\D_{ \eta }$ is generated by open sets of the form
\begin{equation*}
U_{ f, \varepsilon }^{ \phi } := \{ \phi' : \| P_{ \Delta }^{ \phi' } f( \x )- P_{ \Delta }^{ \phi } f( \x ) \|_{ 1, \nu } < \varepsilon \},
\end{equation*}
for any $\phi \in \D_{ \eta }$, $\varepsilon > 0$ and $f \in C_b( \Sim_d )$, the set of continuous and  bounded functions on $\Sim_d$, where $\| \cdot \|_{ p, \nu }$ is the $L^p( \Sim_d, \nu )$-norm.
The Lebesgue measure is meant whenever no measure is specified.
\end{defn}
Lemma 1 of \citep{Koskela17} yields that the topology generated by $U_{f, \varepsilon}^{ \phi }$ is Hausdorff, and hence separates points.
\vskip 11pt
\begin{defn}\label{def_cons}
Let $\n_0, \ldots, \n_m$ denote $m + 1$ samples observed at times $0 , \Delta, \ldots, \Delta m$  from a stationary $\Lambda$-coalescent, with each sample being of size $n \in \N$.
See e.g.~\citep{Beaumont03} for details of how temporally structured samples can be generated.
\emph{Weak posterior consistency} holds if $Q( U_{ \phi_0 }^c | \n_0, \ldots, \n_m ) \rightarrow 0$ $\Prb^{ \phi_0 }$-a.s.~as $n, m \rightarrow \infty$, where $U_{ \phi_0 }$ is any open neighbourhood of $\phi_0 \in \D_{ \eta }$.
\end{defn}
\vskip 11pt
\begin{thm}\label{ts_consistency}
Let $\n_0, \ldots, \n_m$ be as in Definition \ref{def_cons} and $\x_0, \ldots, \x_m$ denote the observed limiting type frequencies as $n \rightarrow \infty$, i.e.~$\x_i = \lim_{ n \rightarrow \infty } \n_i / n$.
Suppose that $\D_{ \eta }$, the support of $Q$, satisfies the conditions of Definition \ref{prior_supp}.
Furthermore, for any $\varepsilon > 0$ and any $\phi_0 \in \D_{ \eta }$ suppose that
\begin{equation}\label{kl_cond}
Q\left( \phi \in \D_{ \eta }:  \int_{ \eta }^1 \left\{ \Big| \log\left( \frac{ \phi_0( r ) }{ \phi( r ) } \right) \Big| + \Big| \frac{ \phi_0( r ) }{ \phi( r ) } - 1 \Big| \right\} r^{ -2 } \phi_0( r ) dr < \varepsilon \right) > 0.
\end{equation}
Then weak posterior consistency holds for $Q$ on $\D_{ \eta }$.
\end{thm}
\vskip 11pt
\begin{rmk}
A similar result for jump diffusions with unit diffusion coefficient was established in \citep[Theorem 1]{Koskela17}, and our proof will follow a similar structure.
Before presenting the proof, let us highlight how the present result differs from the jump diffusion case.
Both proofs of consistency require verification of a Kullback-Leibler condition for the prior, 
and uniform equicontinuity of the family of semigroups corresponding to densities supported by the prior. 
The former result is immediate by the same argument used to prove \citep[Lemma 2]{Koskela17}, whose statement is provided below in Lemma \ref{kl_lemma} in the interest of a self-contained proof.
The latter, Lemma \ref{unif_equicontinuity} below, is different to its counterpart, \citep[Lemma 3]{Koskela17}, which relies on positive definiteness of the diffusion coefficient.
\end{rmk}
\begin{proof}
For fixed $m \in \N$, the same argument used to prove Theorem \ref{consistency} yields that the following convergence holds $\Prb^{ \phi_0 }$-a.s.~as $n \rightarrow \infty$:
\begin{equation*}
\lim_{ n \rightarrow \infty } Q( d\phi | \n_0, \ldots, \n_m ) \propto \pi^{ \phi }( \x_0 ) \prod_{ i = 1 }^m p_{ \Delta }^{ \phi }( \x_{ i - 1 }, \x_i ) Q( d\phi ).
\end{equation*}
Hence it is sufficient to establish posterior consistency for $m + 1$ exact observations from a stationary $\Lambda$-Fleming-Viot process as $m \rightarrow \infty$.
We achieve this by adapting the proof of \citep[Theorem 1]{Koskela17}, which entails verifying two conditions. 
The first is that the prior places sufficient mass in Kullback-Leibler neighbourhoods of $\phi_0$, i.e.~that for any $\varepsilon > 0$ we have
\begin{equation}\label{kl_property}
Q( \phi \in \D_{ \eta } : \KL( \phi_0, \phi ) < \varepsilon ) > 0,
\end{equation} 
where $\KL$ is Kullback-Leibler divergence between $p_{ \Delta }^{ \phi_0 }$ and $p_{ \Delta }^{ \phi }$:
\begin{equation*}
\KL( \phi_0, \phi ) := \int_{ \Sim_d } \int_{ \Sim_d } \log\left( \frac{ p_{ \Delta }^{ \phi_0 }( \x, \y ) }{ p_{ \Delta }^{ \phi }( \x, \y ) } \right) p_{ \Delta }^{ \phi_0 }( \x, \y ) \pi^{ \phi_0 }( \x ) d\y d\x.
\end{equation*}
The second is establishing uniform equicontinuity of $\{ P_{ \Delta }^{ \phi } : \phi \in \D_{ \eta } \}$: 
\begin{equation*}
\sup_{ \phi \in \D_{ \eta } } \sup_{ \| \x - \y \|_2 < \delta } | P_{ \Delta }^{ \phi } f( \x ) - P_{ \Delta }^{ \phi } f( \y ) | < \varepsilon
\end{equation*}
for each $\Delta > 0$ and $f \in C_b( \Sim_d )$.

Condition \eqref{kl_property} follows from a straightforward modification of \citep[Lemma 2]{Koskela17}.
A statement of this result, adapted to the present context and notation, is provided below.
Its proof follows the same structure that of as \citep[Lemma 2]{Koskela17}, and is omitted.

\begin{lem}\label{kl_lemma}
Condition \eqref{kl_cond} implies condition \eqref{kl_property} for any $\varepsilon > 0$.
\end{lem}

It remains to establish uniform equicontinuity on the semigroup $\{ P_{ \Delta }^{ \phi } f : \phi \in \D_{ \eta } \}$ for $f \in C_b( \Sim_d )$.
This can be done by verifying the hypotheses of \citep[Lemma 3]{Koskela17}, which we will do below.

\begin{lem}\label{unif_equicontinuity}
For each $\Delta > 0$ and $f \in C_b( \Sim_d )$, the collection $\{ P_{ \Delta }^{ \phi } f : \phi \in \D_{ \eta } \}$ is uniformly equicontinuous: for any $\varepsilon > 0$ there exists $\delta := \delta( \varepsilon, f, \Delta ) > 0$ such that 
\begin{equation*}
\sup_{ \phi \in \D_{ \eta } } \sup_{ \| \x - \y \|_2 < \delta } | P_{ \Delta }^{ \phi } f( \x ) - P_{ \Delta }^{ \phi } f( \y ) | < \varepsilon.
\end{equation*}
\end{lem}
\begin{proof}
We begin by showing that the required uniform equicontinuity is true for $f \in \operatorname{Lip}(\Sim_d)$, the set of Lipschitz functions on $\Sim_d$.

By \citep[Proposition 2.2, in particular equation (2.2)]{Wang10}, a sufficient condition for equicontinuity for a fixed $\phi \in \D_{ \eta }$ is that for $\| \x - \y \|_2 \leq \delta$ we have
\begin{equation}\label{lip_gen}
\frac{ \theta ( \x - \y )^T ( M - \mathbb{ I }_d ) ( \x - \y ) }{ \| \x - \y \|_2 ( 1 + \| \x - \y \|_2 ) } + \int_{ \eta }^1 \sum_{ i = 1 }^d | x_i - y_i | r^{ -2 } \phi( r ) dr \leq C_{ \delta } \| \x - \y \|_2,
\end{equation}
where $\mathbb{ I }_d$ is the $d \times d$ identity matrix.
Now the first term is trivially bounded by $\theta \| M - \mathbb{ I }_d \|_2 \| \x - \y \|_2$, and the second by $\eta^{ -2 } \sqrt{ d } \| \x - \y \|_2$ using the Cauchy-Schwarz inequality.
Still by \citep[Proposition 2.2]{Wang10}, for $f \in \operatorname{Lip}( \Sim_d )$ and $\| \x - \y \|_2 \leq \delta$ we have the bound
\begin{equation*}
| P_{ \Delta }^{ \phi } f( \x ) - P_{ \Delta }^{ \phi } f( \y ) | \leq ( 2 + \delta^{ -1 } ) e^{ C_{ \delta } \Delta } C( f ) \| \x - \y \|_2,
\end{equation*}
where $C( f )$ is a constant depending only on $f$.
Uniformity in $\phi \in \D_{ \eta }$ now follows from the fact that the constant $C_{ \delta }$ in \eqref{lip_gen} is independent of $\phi$.

Now consider a general test function $f \in C_b( \Sim_d )$.
The simplex $\Sim_d$ is compact, meaning that Lipschitz functions are dense in $C_b( \Sim_d )$.
Thus for any $\beta > 0$, there exists $g \in \operatorname{ Lip }( \Sim_d )$ be such that $\| g - f \|_{ \infty } < \beta$.
The triangle inequality then yields the elementary bound
\begin{align*}
&\sup_{ \phi \in \D_{ \eta } } \sup_{ \| \x - \y \|_2 < \delta } | P_{ \Delta }^{ \phi } f( \x ) - P_{ \Delta }^{ \phi } f( \y ) | \\
&\leq \sup_{ \phi \in \D_{ \eta } } \sup_{ \| \x - \y \|_2 < \delta } \left\{ | P_{ \Delta }^{ \phi } f( \x ) - P_{ \Delta }^{ \phi } g( \x ) | + | P_{ \Delta }^{ \phi } f( \y ) - P_{ \Delta }^{ \phi } g( \y ) | + | P_{ \Delta }^{ \phi } g( \x ) - P_{ \Delta }^{ \phi } g( \y ) |\right\}.
\end{align*}
Now, the first term on the R.H.S.~can be bounded by
\begin{equation*}
| P_{ \Delta }^{ \phi } f( \x ) + P_{ \Delta }^{ \phi } g( \x ) | = \E_{ \x }^{ \phi }[ f( X_t ) - g( X_t ) ] \leq \E_{ \x }^{ \phi }[ | f( X_t ) - g( X_t ) | ] < \beta
\end{equation*}
by construction of $g$.
The second term is bounded analogously.
The last term can be made arbitrarily small by choice of sufficiently small $\delta$.
For fixed $f$, all three bounds are uniform in $\phi$, which concludes the proof.
\end{proof}
The remainder of the proof follows as in \citep{Koskela17}.
It suffices to show that for $f \in C_b( \Sim_d )$ and $B := \{ \phi \in \D_{ \eta } : \| P_{ \Delta }^{ \phi } f - P_{ \Delta }^{ \phi_0 } f \|_{ 1, \nu } > \varepsilon \}$ we have $Q( B | \x_0, \ldots, \x_m ) \rightarrow 0$ $\Prb^{ \phi_0 }$-almost surely.
To that end we fix $f \in \operatorname{Lip}( \Sim_d )$, $\varepsilon > 0$ and thus the set $B$.
Condition \eqref{kl_property} implies that \citep[Lemma 5.2]{vanDerMeulen13} holds, so that if for a measurable collection of subsets $C_m \subset \D_{ \eta }$ there exists $c > 0$ such that
\begin{equation*}
e^{ m c } \int_{ C_m } \pi^{ \phi }( \x_0 ) \prod_{ i = 1 }^m p_{ \Delta }^{ \phi }( \x_{ i - 1 }, \x_i ) Q( d\phi ) \rightarrow 0
\end{equation*}
$\Prb^{ \phi_0 }$-almost surely, then $Q( C_m | \x_0, \ldots, \x_m ) \rightarrow 0$ $\Prb^{ \phi_0 }$-almost surely as well.
Likewise, Lemma \ref{unif_equicontinuity} implies \citep[Lemma 5.3]{vanDerMeulen13}: there exists a compact subset $K \subset \Sim_d$, $N \in \N$ and compact, connected sets $I_1, \ldots, I_N$ that cover $K$ such that 
\begin{equation*}
B \subset \bigcup_{ j = 1 }^N B_j^+ \cup \bigcup_{ j = 1 }^N B_j^-,
\end{equation*}
where
\begin{equation*}
B_j^{ \pm } := \left\{ \phi \in \D_{ \eta } : P_{ \Delta }^{ \phi } f( \x ) - P_{ \Delta }^{ \phi_0 } f( \x ) > \pm \frac{ \varepsilon }{ 4 \nu( K ) } \text{ for every } \x \in I_j \right\}.
\end{equation*}
Thus it is only necessary to show $Q( B_j^{ \pm } | \x_0, \ldots, \x_m ) \rightarrow 0$ $\Prb^{ \phi_0 }$-almost surely.
Define the stochastic process
\begin{equation*}
D_m := \left( \int_{ B_j^+ } \pi^{ \phi }( \x_0 ) \prod_{ i = 1 }^m p_{ \Delta }^{ \phi }( \x_{ i - 1 }, \x_i ) Q( d\phi ) \right)^{ 1 / 2 }.
\end{equation*}
Now $D_m \rightarrow 0$ exponentially fast as $m \rightarrow \infty$ by an argument identical to that used to prove \citep[Theorem 3.5 ]{vanDerMeulen13}.
The same is also true of the analogous stochastic process defined by integrating over $B_j^-$, which completes the proof.
\end{proof}

In the next section we show that given a data set of size $n$, it is natural to infer the first $n - 2$ moments of $\Lambda$ because they fully capture the signal in the data set.
Example \ref{prior_example} at the end of Section \ref{parameters} provides a family of priors which satisfy the hypotheses of Theorem \ref{ts_consistency}, and whose support can be chosen to contain arbitrarily close approximations to truncated moment sequences of any $\Lambda \in \M_1( [ 0, 1 ] )$.
\vskip 11pt
\begin{rmk}
The hypotheses of Theorem \ref{ts_consistency} are strong, and thus it would be desirable to obtain a posterior contraction rate in addition to just consistency.
In fact, methods akin to that employed in the proof have been extended to provide rates for compound Poisson processes \citep{Gugushvili15} and scalar diffusions on compact intervals \citep{Nickl15}.
However, extending either approach to our setting would require bounds of the form 
\begin{align*}
\| \pi^{ \phi } - \pi^{ \phi_0 } \|_2 \leq C n^{ - \beta }, & & \| \pi^{ \phi } / \pi^{ \phi_0 } - 1 \|_2 \leq \tilde{ C } n^{ - \tilde{ \beta } }
\end{align*}
for some constants $\beta, \tilde{ \beta }, C, \tilde{ C } > 0$.
Since the $\Lambda$-Fleming-Viot stationary density is intractable in nearly all cases, it does not seem possible to extend our approach to obtain rates of posterior consistency.
\end{rmk}

\section{Parametrisation by finitely many moments}\label{parameters}

Consider a set  of type frequencies $\n \in \N^d$ of size $n := \sum_{ i = 1 }^d n_i$ generated by a $\Lambda$-coalescent with finite alleles mutation started from $\psi_n$.
\vskip 11pt
\begin{lem}\label{unidentifiability}
The likelihood satisfies $\mathbf{ P }_{ n }^{ \Lambda }( \n ) = \mathbf{ P }_{ n }^{ \lambda_3, \lambda_4, \ldots, \lambda_n }( \n )$.
That is, $\Lambda$ is conditionally independent of $\n$ given $\{ \lambda_k \}_{ k = 3 }^n$.
\end{lem}
\begin{proof}
Let $- q_{ n n } = \sum_{ k = 1 }^{ n - 1 } \binom{ n }{ n - k + 1 } \lambda_{ n, n - k + 1 }$ be the total merger rate of the $\Lambda$-coalescent with $n$ blocks.
It is well known that the $\Lambda$-coalescent likelihood is the unique solution to the recursion \citep{Mohle06, Birkner08, Birkner09b}:
\begin{align}
\mathbf{ P }_{ n }^{ \Lambda }( \n ) = &\frac{ \theta }{ n \theta - q_{ n n } } \sum_{ i, j = 1 }^d ( n_j - 1 + \delta_{ i j } ) M_{ j i } \mathbf{ P }_{ n }^{ \Lambda }( \n - \e_i + \e_j ) \nonumber \\
&+ \frac{ 1 }{ n \theta - q_{ n n } } \sum_{ i : n_i \geq 2 } \sum_{ k = 2 }^{ n_i } \binom{ n }{ k } \lambda_{ n, k } \frac{ n_i - k + 1 }{ n - k + 1 } \mathbf{ P }_{ n - k + 1 }^{ \Lambda }( \n - ( k - 1 )\e_i ), \label{gt}
\end{align}
with boundary condition $\mathbf{ P }_1^{ \Lambda }( \e_i ) = m( i )$.
Repeated application of this recursion yields a closed system of linear equations for the likelihood because all sample sizes on the R.H.S.~are equal to or smaller than the one on the left hand side.
This system is far too large to solve for all but very small sample sizes, but it is clear that the solution can depend on $\Lambda$ only through the polynomial moments $\{ \lambda_{ m, k } \}_{ k \leq m = 2 }^n$.

Polynomial moments can be written as a linear combination of monomial moments:
\begin{equation}
\lambda_{ m, k } = \sum_{ j = 0 }^{ m - k } \binom{ m - k }{ j } ( -1 )^j \lambda_{ k + j }, \label{moment_consistency}
\end{equation}
meaning that only the monomial moments $\{ \lambda_k \}_{ k = 2 }^n$ are required. 
Since $\lambda_2 = \int_0^1 \Lambda( dx ) = 1$, the moments $\{ \lambda_k \}_{ k=3 }^n$ are sufficient.
\end{proof}

Motivated by Lemma \ref{unidentifiability} we make the following definition:
\vskip 11pt
\begin{defn}
Let $\sim_n$ be the equivalence relation on $\M_1( [ 0, 1 ] )$ defined via
\begin{equation*}
\Lambda_1 \sim_n \Lambda_2 \text{ if } \lambda_k^{ ( 1 ) } = \lambda_k^{ ( 2 ) } \text{ for } k \in \{ 3, \ldots, n \}
\end{equation*}
where $\lambda_k^{ ( i ) } := \int_0^1 x^{ k - 2 } \Lambda_i( dx )$.
We call the equivalence classes of $\sim_n$ moment classes of order $n$.
\end{defn}
\vskip 11pt
In view of Lemma \ref{unidentifiability} it is natural to consider the problem of inferring $\Lambda$ from $\n$ in the quotient space $\M_1( [ 0, 1 ] ) / \sim_n$, not in $\M_1( [ 0, 1 ] )$.
Moreover, requiring all linear combinations of the form \eqref{moment_consistency} to be non-negative guarantees a unique solution to the Hausdorff moment problem, so that each completely monotonic moment sequence bounded by 1 corresponds to some $\Lambda \in \M_1( [ 0, 1 ] )$.
Hence we parametrise the space $\M_1( [ 0, 1 ] ) / \sim_n$ by truncated, completely monotonic moment sequences of length $n - 2$ with leading term $\lambda_3 \leq 1$.
This approach yields a compact, finite-dimensional parameter space which nevertheless captures all the signal in the data.
Table \ref{mom_seq} lists some moment sequences corresponding to popular families of $\Lambda$-measures.
\begin{table}[!ht]
\centering
\begin{tabular}{c|cccccc}
$\Lambda$ & $\delta_0$ & $\delta_1$ & $\operatorname{Beta}( 2 - \alpha, \alpha )$ & $U( 0, 1 )$ & $\frac{ 2 }{ 2 + \psi^2 } \delta_0 + \frac{ \psi^2 }{ 2 + \psi^2 } \delta_{ \psi }$ & $c \delta_0 + \frac{ ( 1 - c ) }{ 2 } r dr$\\
\hline
$\lambda_k $ & 0 & 1 & $\frac{ ( 2 - \alpha )_{ k - 2 } }{ ( 2 )_{ k - 2 } }$ & $\frac{ 1 }{ k - 1 }$ & $ \frac{ \psi^k }{ 2 + \psi^2 } $ & $\frac{ 1 - c }{ 2 k }$
\end{tabular}
\caption{Moment sequences of particular $\Lambda$-coalescents. Here $( a )_{ k } := a ( a + 1 ) \ldots ( a + k - 1 )$ denotes the rising factorial.}
\label{mom_seq}
\end{table}

Naturally, the prior $Q$ ought to be chosen to yield a tractable push-forward prior on truncated moment sequences.
These push-forward priors inherit posterior consistency whenever $Q$ satisfies the conditions of Theorem \ref{ts_consistency} because truncated moment sequences can be written as bounded functionals of $\Lambda$.

\section{Prior distributions}\label{priors}

In this section we provide an example family of priors which satisfies the consistency criteria of Theorem \ref{ts_consistency} and has tractable push-forward distributions on truncated moment sequences.
\vskip 11pt
\begin{defn}
Let $Q \in \M_1( \M_1( [ 0, 1 ] ) )$ be a prior distribution for $\Lambda$.
Then the moments $\{ \lambda_k \}_{ k = 3 }^n$ have joint prior $Q_n$ on the space of completely monotonic sequences of length $n - 2$ given by
\begin{equation}\label{moment_prior}
Q_n( \lambda_3 \in dy_3, \ldots, \lambda_n \in dy_n ) := \int_{ \M_1( [ 0, 1 ] ) } \prod_{ k = 3 }^n \mathds{ 1 }_{ \{ dy_k \} }\left( \int_{ ( 0, 1 ] } r^{ k - 2 } \Lambda( dr ) \right) Q( d\Lambda ).
\end{equation}
\end{defn}
The prior $Q$ should to be chosen such that the R.H.S.~of \eqref{moment_prior} is tractable, and the following example illustrates that such a choice is possible.
\begin{ex}\label{prior_example}
\vskip 11pt
Fix $\eta > 0$ and $\alpha \in \M( [ \eta, 1 ] )$ with finite mass and a strictly positive Lebesque density $\alpha( r )$.
Suppose $\D_{ \eta }$ satisfies the conditions of Definition \ref{prior_supp}, and in addition that every $\phi \in \D_{ \eta }$ is continuous.
Let $R( d\tau )$ be a probability measure on $( 0, \infty )$ placing positive mass in all non-empty open sets.
For $x \in [ \eta, 1 ]$ and $\tau > 0$ let 
\begin{equation*}
q_{ x, \tau }( r ) := \frac{ \mathds{ 1 }_{ [ \eta, 1 ] }( r - x ) h_{ \tau }( r - x ) }{ h_{ \tau }( [ \eta, 1 ] ) },
\end{equation*}
where $h_{ \tau }$ is the Gaussian density on $\R$ with mean 0 and variance $\tau^{ -1 }$.

Let $DP( \alpha )$ be the law of a Dirichlet process centred on $\alpha$ \citep{Ferguson73} and let $Q$ be given by the Dirichlet process mixture distribution \citep{Lo84} with mixing distribution $DP( \alpha ) \otimes R$ and mixture components $q_{ x, \tau }$.
In other words, let $F \sim DP( \alpha )$ and $F( r_1 ), F( r_2 ), \ldots$ be the weights of ordered atoms of $F$.
Let $\{ \tau_i \}_{ i =1 }^{ \infty }$ be i.i.d.~draws from $R$. 
Then a draw from $Q$ is given by
\begin{equation*}
\phi( r ) | F, \{ \tau_i \}_{ i = 1 }^{ \infty } = \sum_{ i = 1 }^\infty F( r_i ) q_{ r_i, \tau_i }( r ).
\end{equation*}

The prior $Q$ places full mass on equivalent densities bounded from above and away from 0 by construction.
We assume $Q$ satisfies Lemma \ref{phi_ident}, at which point it remains to check \eqref{kl_cond} to verify $Q$ has posterior consistency.
Theorem 1 of \citep{Bhattacharya12} yields that for any $\delta > 0$ the prior places positive mass in all open balls: $Q( \phi \in \D_{ \eta } : \| \phi - \phi_0 \|_{ \infty } < \delta ) > 0$ for any $\phi_0 \in \D_{ \eta }$ under the above assumptions on $\alpha$, $R$ and $( q_{ x, \tau } )_{ x \in [ \eta, 1 ], \tau > 0 }$.
Now fix $\varepsilon > 0$, $\phi_0 \in \D_{ \eta }$, as well as $0 < \delta < c$ such that
\begin{equation*}
\log\left( \frac{ c }{ c - \delta } \right) \vee \Big| \log\left( \frac{ c }{ c + \delta } \right) \Big| + \frac{ \delta }{ c - \delta } < \eta^2 \varepsilon.
\end{equation*}
Then
\begin{align*}
&Q\left( \phi \in \D_{ \eta } : \int_{ \eta }^1 \left\{ \Big| \log\left( \frac{ \phi_0( r ) }{ \phi( r ) } \right) \Big| + \Big| \frac{ \phi_0( r ) }{ \phi( r ) } - 1 \Big| \right\} r^{ -2 } \phi_0( r ) dr < \varepsilon \right) \\
&\geq Q( \phi \in \D_{ \eta } : \| \phi - \phi_0 \|_{ \infty } < \delta ) > 0,
\end{align*}
because for any $\phi$ satisfying $\| \phi - \phi_0 \|_{ \infty } < \delta$ we have
\begin{align*}
&\int_{ \eta }^1 \left\{ \Bigg| \log\left( \frac{ \phi_0( r ) }{ \phi( r ) } \right) \Bigg| + \Bigg| \frac{ \phi_0( r ) }{ \phi( r ) } - 1 \Bigg| \right\} r^{ -2 } \phi_0( r ) dr \\
&\leq \int_{ \eta }^1 \left\{ \log\left( \frac{ \phi_0( r ) }{ \phi_0( r ) - \delta } \right) \vee \Big| \log\left( \frac{ \phi_0( r ) }{ \phi_0( r ) + \delta } \right) \Big| + \Big| \frac{ \phi_0( r ) }{ \phi_0( r ) - \delta } - 1 \Big| \right\} r^{ -2 } \phi_0( r ) dr \\
&\leq \left( \log\left( \frac{ c }{ c - \delta } \right) \vee \Big| \log\left( \frac{ c }{ c + \delta } \right) \Big| + \frac{ \delta }{ c - \delta } \right) \eta^{-2 } < \varepsilon.
\end{align*}

We now use the machinery of \citet{Regazzini02} to give an explicit system of equations for the distribution function of $Q_n$ under this choice of $Q$.
Define the family of functions 
\begin{equation*}
g_p( x ) := \int_{ \eta }^1 r^p ( c + ( 1 - c( 1 - \eta ) ) q_{ x, \tau }( r ) ) dr
\end{equation*}
for $p \in \N$ and $x \in [ \eta, 1 ]$, as well as the vectors $\g_n( x ) := ( g_1( x ), g_2( x ), \ldots, g_n( x ) )$ and $\s_n := ( s_1, \ldots, s_n ) \in \R^n$.
For brevity, for a measure $\nu$ and a function $f$ let $\nu( f ) := \int f d\nu$ whenever the integral exists.

Let $\gamma_{ \alpha }$ be a Gamma random measure with parameter $\alpha$, that is, a random finite measure on $[ \eta, 1 ]$ such that for any measurable partition $\{ A_1, \ldots, A_n \}$ the random variables $( \gamma_{ \alpha }( A_1 ), \ldots, \gamma_{ \alpha }( A_n ) )$ are independent and gamma distributed with common scale parameter 1 and respective shape parameters $\alpha( A_k )$.
Let 
\begin{equation*}
h_n( \s_n; \g_n; \alpha ) := \E\left[  \exp\left( i \s_n \cdot \gamma_{ \alpha }( \g_n ) \right) \right]
\end{equation*}
be the characteristic function of $\gamma_{ \alpha }( \g_n ) := ( \gamma_{ \alpha }( g_1 ), \ldots, \gamma_{ \alpha }( g_n ) )$.
Note that $h_n( \s_n; \g_n; \alpha ) = h_n( \mathbf{ 1 }; \s_n \cdot \g_n; \alpha )$ and by \citep[Proposition 10]{Regazzini02}
\begin{equation}\label{h_char}
h_n( \s_n; \g_n; \alpha ) = \exp\left( - \int_0^1 \log( 1 - i \s_n \cdot \g_n ) d\alpha \right).
\end{equation}
Now let $F_n( \bm{\sigma}, \g_n, \alpha )$ be the joint distribution function of $F( \g_n ) := ( F( g_1 ), \ldots, F( g_n ) )$ under $DP( \alpha )$.
The following trick was introduced in \citep[equation (2.9)]{Hannum81}:
\begin{equation*}
F_n( \bm{ \sigma }, \g_n, \alpha ) = F_n( \mathbf{ 0 }, \gamma_{ \alpha }( \g_n - \bm{ \sigma } ), \alpha )
\end{equation*}
for any $\bm{ \sigma } \in \R^n$, so that it is sufficient to invert $h_n$ at the origin to obtain $F_n$.
This can be done using the multidimensional version of the Gurland inversion formula \citep[Theorem 3]{Gurland48}:

Let $C_0, \ldots, C_n \in \R^{ n + 1 }$ solve 
\begin{align}
C_n &= -1 \nonumber\\
\sum_{ k = 0 }^{ n - r - 1 } \binom{ n - r }{ k } C_{ r + k } &= 1 \text{ for } r \in \{ 0, \ldots, n - 1 \} \label{c_ks} .
\end{align}
Then
\begin{align*}
( -1 )^{ n + 1 } 2^n &F_n( \bm{ \sigma }, \g_n, \alpha ) = C_0 \\
&+ \sum_{ k = 1 }^n \frac{ C_k }{ ( \pi i )^k } \sum_{ 1 \leq j_1 < \ldots < j_k \leq n } \int_0^{ \infty } \ldots \int_0^{ \infty } \frac{ h_k( \s_k; g_{ j_1 } - \sigma_{ j_1 }, \ldots, g_{ j_k } - \sigma_{ j_k }; \alpha ) }{ s_1 \times \cdots \times s_k } d\s_k.
\end{align*}
The characteristic functions $h_k$ and the constants $C_k$ can be computed from \eqref{h_char} and \eqref{c_ks} respectively, so that the R.H.S.~can be evaluated numerically for practical applications.
Numerical methods are discussed in \citep[Section 6]{Regazzini02}.

Finally, we demonstrate that the restrictive assumptions of Theorem \ref{ts_consistency} still allow inference for broad classes of moment sequences with arbitrarily small approximation errors.
Let $\beta( r )$ be any non-negative probability density on $[ 0, 1 ]$, and define the truncation $\bar{ \beta }( r ) := \kappa( \eta, c, K )^{ -1 } ( \beta( r ) \vee c \wedge K )$, where $\kappa$ is the normalising constant
\begin{equation*}
\kappa( \eta, c, K ) := \int_{ \eta }^1 \beta( r ) - ( \beta( r ) - K )^+ + ( c - \beta( r ) )^+ dr.
\end{equation*}
Note that $Q( \phi \in \D_{ \eta } : \| \phi - \bar{ \beta } \|_{ \infty } < \delta ) > 0$ for any $\delta > 0$, and fix such a $\phi$.
Now consider the error on the $k^{\text{th}}$ moment:
\begin{align*}
&\Big| \int_0^1 r^k \beta( r ) dr - \int_{ \eta }^1 r^k \phi( r )dr \Big| \\
&\leq \beta( ( 0, \eta ) ) + ( 1 - \eta ) \delta +  \frac{ ( \kappa( \eta, c, K ) - 1 ) \beta( [ \eta, 1 ] ) + c( 1 - \eta ) }{ \kappa( \eta, c, K ) } + \int_{ \eta }^1 \frac{ ( \beta( r ) - K )^+ }{ \kappa( \eta, c, K ) } dr.
\end{align*}
Each term on the R.H.S.~can be made small by choosing $\eta$, $c$ and $\delta$ sufficiently small, and $K$ sufficiently large because $\kappa \rightarrow 1$ as $\eta \rightarrow 0$, $c \rightarrow 0$ and $K \rightarrow \infty$, and 
\begin{equation*}
\int_{ \eta }^1 ( \beta( r ) - K )^+ dr = 1 - \int_{ \eta }^1 \beta( r ) \wedge K dr \rightarrow \beta( ( 0, \eta ) )
\end{equation*}
as $K \rightarrow \infty$ by the Monotone Convergence Theorem.
A further approximation step also enables consideration of atoms by choosing $\beta$ which places all of its mass in neighbourhoods of the desired locations for atoms.
Hence it is possible to ensure the support of $Q$ extends arbitrarily close to any desired moment sequences despite the restrictive assumptions on $\D_{ \eta }$ in Theorem \ref{ts_consistency}.
\end{ex}

\section{Robust bounds on functionals of \texorpdfstring{$\Lambda$}{Lambda}}\label{positive_results}

Having established consistency criteria for the posterior and a finite parametrisation via $n - 2$ leading moments, we now turn to what can be said about $\Lambda$ based on inferring these moments.
It would be ideal if the diameter of moment classes shrunk with increasing $n$, as then it would be possible to fix a representative $\Lambda \in \M_1( [ 0, 1 ] )$ with specified $n - 2$ leading moments and control the remaining within-moment-class error.
In Theorem \ref{moment_class_thm} we show that such shrinking does not happen, and devote the remainder of the section to presenting quantities which can be controlled based on $n - 2$ moments alone.
We begin by recalling some standard results from the theory of orthogonal polynomials.
\vskip 11pt
\begin{defn}
Suppose $n$ is odd.
Let $m := \frac{ n - 3 }{ 2 }$ and $\{ \phi_k \}_{ k = 0 }^{ m }$ be the first $m + 1$ $\Lambda$-orthogonal polynomials.
Let $\{ \xi_k \}_{ k = 1 }^m$ be the zeros of $\phi_m$.
\vskip 11pt
\end{defn}
\begin{rmk}
It is a standard result that $\{ \phi_k \}_{ k = 0 }^{ m - 1 }$ and $\{ \xi_k \}_{ k = 1 }^m$ are constant within moment classes of order at least $n$.
\end{rmk}
The following bounds on $\Lambda$ in terms of its leading $n - 2$ moments are classical:
\vskip 11pt
\begin{lem}[Chebyshev-Markov-Stieltjes (CMS) inequalities]
Define
\begin{equation*}
\rho_{ m - 1 }( z ) := \left( \sum_{ k = 0 }^{ m - 1 } | \phi_k( z ) |^2 \right)^{ -1 }.
\end{equation*}
Then the following inequalities are sharp:
\begin{equation*}
\Lambda( [ 0, \xi_j ] ) \leq \sum_{ k = 1 }^j \rho_{ m - 1 }( \xi_k ) \leq \Lambda( [ 0, \xi_{ j + 1 } ) ) \text{ for } j = \{ 1, \ldots, m \},
\end{equation*}
where $\xi_{ m + 1 } := 1$.
\end{lem}
We are now in a position to prove the following theorem:
\vskip 11pt
\begin{thm}\label{moment_class_thm}
For any $n \in \N$ and any completely monotonic sequence of moments $\{ \lambda_{ k, k } \}_{ k = 3 }^n$ with $\lambda_{3,3} \leq 1$ there exist uncountably many measures $\Lambda \in \M_1( [ 0, 1 ] )$, all with leading moments $\{ \lambda_{ k, k } \}_{ k = 3}^n$ and all satisfying the CMS inequalities, such that for any pair $\Lambda_x$ and $\Lambda_y$ of them we have $d_{ TV }( \Lambda_x, \Lambda_y ) = 2$, where $d_{ TV }$ is the total variation distance.
\end{thm}
\begin{proof}
It will be convenient to write the CMS inequalities in the following, equivalent form:
\begin{align*}
0 &\leq \Lambda( [ 0, \xi_1) ) \leq \rho_{ m - 1 }( \xi_1 ) \\
0 &\leq \Lambda( [ \xi_j, \xi_{ j + 1 } ) ) \leq \rho_{ m - 1 }( \xi_j ) + \rho_{ m - 1 }( \xi_{ j + 1 } ) \text{ for } j \in \{ 1, \ldots, m - 1 \} \\
0 &\leq \Lambda( [ \xi_m, 1 ] ) \leq \rho_{ m - 1 }( \xi_m )
\end{align*}
where the last inequality follows from the fact that $\sum_{ k = 1 }^m \rho_{ m - 1 }( \xi_k ) = 1$.
This equality holds because $\sum_{ k = 1 }^m \rho_{ m - 1 }( \xi_k )$ is the sum of all order $m$ Gauss quadrature weights, or equivalently the quadrature applied to the constant function 1, which is a polynomial of degree 0.
The equality follows by recalling that Gauss quadrature is exact for polynomials of order up to $2 m - 1$.

Now let the measures $\Lambda_x$ and $\Lambda_y$ be described by sequences of $( m + 1 )$ weights $( x_0, x_1, \ldots, x_m )$ and $( y_0, y_1, \ldots, y_m )$, with the $j^{\text{th}}$ weight denoting the mass that the corresponding measure places in the interval $[ \xi_j, \xi_{ j + 1 } )$ (with obvious adjustments for the right hand boundary terms).

For brevity let $\zeta_j := \rho_{ m - 1 }( \xi_j )$.
Suppose first that $m$ is odd, and let the vectors of weights be given as
\begin{align*}
( x_0, x_1, x_2, x_3, x_4, \ldots, x_{ m - 1 }, x_m ) &= ( \zeta_1, 0, \zeta_2 + \zeta_3, 0, \zeta_4 + \zeta_5, \ldots, 0, \zeta_m ) \\
( y_0, y_1, y_2, y_3, y_4, \ldots, y_{ m - 1 }, y_m ) &= ( 0, \zeta_1 + \zeta_2, 0, \zeta_3 + \zeta_4, 0, \ldots, \zeta_{ m - 1 } + \zeta_m, 0 )
\end{align*}
Both measures have total mass $\sum_{ j = 1 }^m \zeta_j = 1$, and the interlacing masses have no overlap so $d_{ TV }( \Lambda_x, \Lambda_y ) = 2$.
The case where $m$ is even is similar.
\end{proof}
\begin{rmk}
The same result holds in Kullback-Leibler divergence due to Pinsker's inequality: for probability measures $P$ and $Q$ such that $P \sim Q$, and letting $\operatorname{KL}( P, Q ) := \E_P[ \log( dP / dQ ) ]$ we have
\begin{equation*}
d_{ TV }( P, Q ) \leq \sqrt{ \frac{ 1 }{ 2 } \operatorname{KL}( P, Q ) },
\end{equation*}
so that $\operatorname{KL}( \Lambda_x, \Lambda_y ) \geq 8$ for $\Lambda_x$ and $\Lambda_y$ as in Theorem \ref{moment_class_thm}.
\end{rmk}
Despite this seemingly disappointing result, it is possible to make some conclusions about $\Lambda$ based on $n - 2$ moments.
For example, the Kingman hypothesis can be tested in a robust way by checking whether the vector $( 0, 0, \ldots, 0 )$ lies in a desired credible region of  the posterior $Q_n( \cdot | \n )$, and the plausibility of any other $\Lambda$ of interest can be assessed similarly.
More generally, it is possible to maximise/minimise a certain class of functionals subject to moment constraints obtained from a credible region to obtain robust bounds for quantities of interest.
We begin by recalling some relevant definitions.
\vskip 11pt
\begin{defn}\label{cm_def}
Let $m, n \in \N$, and fix $\R$-valued constants $\{ c_k \}_{ k = 1 }^m$, a sequence $\{ i_k \}_{ k = 1 }^m$ of $\{ 3, \ldots, n \}$-valued indices and a binary sequence $\{ j_k \}_{ k = 1 }^m$ of zeros and ones.
Let
\begin{equation}\label{confidence}
\C_m := \left\{ \Lambda \in \M_1( [ 0, 1 ] ) : (-1)^{ j_k } \lambda_{ i_k } \leq c_k \text{ for } k \in \{ 1, \ldots, m \} \right\}
\end{equation}
be a subset of $\M_1( [ 0, 1 ] )$ with leading $n - 2$ moments in a desired region specified by $m$ linear inequalities.
Let $\ext( \C_m )$ be the extremal points in $\C_m$, i.e. those which cannot be written as non-trivial convex combinations of elements in $\C_m$, and
\begin{equation*}
\C_{ D } := \left\{ \nu \in \C_m : \nu = \sum_{ k = 1 }^p w_k \delta_{ x_k } \text{ where } 1 \leq p \leq m + 1, w_k \geq 0, x_k \in [ 0, 1 ] \text{ and } \sum_{ k = 1 }^p w_k = 1 \right\}
\end{equation*}
be the set of discrete probability measures on $[ 0, 1 ]$ with at most $m + 1$ atoms.
\end{defn}
\vskip 11pt
\begin{ex}
The extremal points of $\M_1( [ 0, 1 ] )$ are the Dirac measures:
\begin{equation*}
\ext\left( \M_1( [ 0, 1 ] ) \right) = \{ \delta_x : x \in [ 0, 1 ] \}.
\end{equation*}
\end{ex}

For our purposes $\C_m$ should be thought of as an envelope containing a desired credible region of truncated, completely monotonic moment sequences expressed using finitely many linear constraints.
We postpone discussion of how an approximate credible region can be obtained to the next section, and simply assume one is available.
The importance of Definition \ref{cm_def} is that maxima and minima of certain functionals of $\Lambda$ coincide on $\C_m$ and $\C_D$, and that $\C_D$ is finite-dimensional so that these extrema can be found numerically.
The class of functionals for which this can be done is given below.
\vskip 11pt
\begin{defn}\label{measure_affine}
The functional $F: \C \mapsto \overline{ \R }$ is \emph{measure-affine} if, for every $\nu \in \C$ and $p \in \M_1( \ext( \C ) )$ such that $\nu( E ) = \int_{ \ext( \C ) } \gamma( E ) p( d \gamma )$ for every $E \in \B( [ 0, 1 ] )$, $F$ is $p$-integrable and 
\begin{equation*}
F( \nu ) = \int_{ \ext( \C ) } F( \gamma ) p( d\gamma ).
\end{equation*}
\end{defn}
Intuitively, $\nu$ is a barycentre of $\C$ with weights on extremal points given by $p$, and $F$ is measure-affine if it commutes with the operation of expressing $\nu$ as the weighted sum of extremal points.
If $\C$ consists of finitely many points, this definition coincides with the usual definition of affine functions.

The following two results originate from \citep[Proposition 3.1 and Theorem 3.2]{Winkler88}:
\vskip 11pt
\begin{lem}\label{affine_lemma}
If $q: [ 0, 1 ] \mapsto \overline{ \R }$ is bounded on at least one side then $F: \nu \mapsto \E_{ \nu }[ q ]$ is measure-affine. 
\end{lem}
\vskip 11pt
\begin{lem}\label{optimisation}
Let $\C_m$ be as in Definition \ref{cm_def} and $F: \C_m \mapsto \overline{ \R }$ be measure-affine.
Then
\begin{align}
\inf_{ \nu \in \C_m } F( \nu ) &= \inf_{ \nu \in \C_{ D } } F( \nu ) \label{inf_opt} \\
\sup_{ \nu \in \C_m } F( \nu ) &= \sup_{ \nu \in \C_{ D } } F( \nu ). \label{sup_opt}
\end{align}
\end{lem}
The purpose of Lemma \ref{optimisation} is that the optimisation problems on the R.H.S.~of \eqref{inf_opt} and \eqref{sup_opt} are finite-dimensional and can be solved numerically.
Hence tight bounds for measure-affine functionals $F( \Lambda )$ over credibility regions can be computed in an assumption-free manner.

In order to specify $\C_m$ it remains to be able to approximate the posterior, which we achieve via MCMC.
This will be detailed in the next section.
Before that, we conclude this section with a simple example computation.
\vskip 11pt
\begin{ex}
Suppose a posterior credible region is specified via two linear constraints as
\begin{equation*}
\C_2 = \{ \Lambda \in \M_1( [ 0, 1 ] ) : \lambda_3 \leq 0.5 \text{ and } 0.3 \leq \lambda_4 \},
\end{equation*}
and that the measure-affine functional of interest is the exponential:
\begin{equation*}
F( \Lambda ) := \int_0^1 e^{ -r } \Lambda( dr ).
\end{equation*}
Then the finite dimensional subspace $\C_D \subset \C_2$ consists of discrete probability measures on $[0, 1]$ with at most three atoms:
\begin{equation*}
\C_D = \left\{ \Lambda \in \C_2 : \Lambda = \sum_{ k = 1 }^p w_k \delta_{ r_k } \text{ where } 1 \leq p \leq 3, w_k \geq 0, r_k \in [ 0, 1 ] \text{ and } \sum_{ k = 1 }^p w_k = 1 \right\}.
\end{equation*}
This yields three maximisation/minimisation problems, one corresponding to each number of atoms, though in practice only the largest needs to be solved since the two others can be recovered as special cases.
In this case, the constrained optimisation problem is
\begin{align*}
\text{Maximise/Minimise: } & a e^{ -x } + b e^{ - y } + ( 1 - a - b ) e^{ -z } \\
\text{Subject to: } & a x + b y + ( 1 - a - b ) z \leq 0.5 \\
&-a x^2 - b y^2 - ( 1 - a - b ) z^2 \leq -0.3 \\
&a, b \leq 1 \\
&-a, -b \leq 0 \\
&a + b \leq 1 \\
&x, y, z \leq 1 \\
&-x, -y, -z \leq 0.
\end{align*}
Numerical evaluation in Mathematica yields the bounds $F( \Lambda ) \in ( 0.620, 0.810 )$.
\end{ex}

\section{Simulation study}\label{simulation}

Efficient methods for approximating the $\Lambda$-coalescent likelihood pointwise exist \citep{Birkner11, Koskela15} and can be readily adapted to the form developed by \citet{Beaumont03} for time series data.
These likelihood estimators can then be used in the pseudo-marginal Metropolis-Hastings algorithm \citep{Beaumont03, Andrieu09}, in which the likelihood evaluations required in a standard Metropolis-Hastings algorithm are replaced with unbiased estimators.
The resulting algorithm still targets the correct posterior and inherits the efficient exploration of parameter space of MCMC methods.
Thus it is well-suited to high-dimensional situations with intractable likelihood.

Let $S( n )$ denote the space of completely monotonic sequences of length $n - 2$, and for $\bm{ \lambda } \in S( n )$ let $L( \bm{ \lambda }; \n )$ be the likelihood function and $\widehat{ L }( \bm{ \lambda }; \n )$ be an unbiased estimator.
We recall the pseudo-marginal Metropolis-Hastings algorithm in Algorithm \ref{MCMC} below.

\begin{algorithm}[!ht]
\caption{Pseudo-marginal Metropolis-Hastings for finite moment sequences}
\label{MCMC}
\begin{algorithmic}[1]
\REQUIRE Prior $P_n$, observation $\n$, transition kernel $K: S(n) \times S(n) \mapsto \R_+$ and $N \in \N$.
\STATE Initialise sample $S \gets \emptyset$ and moment sequence $\bm{ \lambda } \gets \bm{ \lambda }_0$. 
\STATE Compute likelihood estimator $\widehat{ L }( \bm{ \lambda }; \n )$.
\FOR{j = 1, \ldots, N}
\STATE Sample $\bm{ \lambda }' \sim K( \bm{ \lambda }, \cdot)$.
\STATE Compute likelihood estimator $\widehat{ L }( \bm{ \lambda }'; \n )$. 
\STATE Set $a \gets 1 \wedge \frac{ K( \bm{ \lambda }', \bm{ \lambda } ) \widehat{ L }( \bm{ \lambda }'; \n ) P_n( \bm{ \lambda }' ) }{ K( \bm{ \lambda }, \bm{ \lambda }' )\widehat{ L }( \bm{ \lambda }; \n ) P_n( \bm{ \lambda } ) }$. 
\STATE Sample $U \sim U(0,1)$.
\IF{$U < a$}
\STATE Set $S \gets S \cup \{ \bm{ \lambda }' \}$, $\widehat{ L }( \bm{ \lambda } ) \gets \widehat{ L }( \bm{ \lambda }' )$ and $\bm{ \lambda } \gets \bm{ \lambda }'$.
\ELSE
\STATE Set $S \gets S \cup \{ \bm{ \lambda } \}$.
\ENDIF
\ENDFOR
\RETURN $S$
\end{algorithmic}
\end{algorithm}

Algorithm \ref{MCMC} returns a sample of moment sequences $S$, whose limiting distribution is the posterior.
A credible region $C$ can be approximated from MCMC output, and used to form $\C_m$ as per \eqref{confidence}.
Measure-affine quantities of interest can then be maximised or minimised using finite computation by making use of Lemma \ref{optimisation}.

By way of demonstration we perform a simulation study on simulated data, focusing on assessing the Kingman hypothesis, $\Lambda = \delta_0$, which can be robustly evaluated based upon whether or not $\lambda_3 = 0$.
The type space consists of 10 binary loci, or $2^{10}$ types, with mutations flipping a uniformly chosen locus, i.e.~the matrix $M$ has 10 identical, non-zero entries in each row corresponding to flipping one locus each.
The total mutation rate is set at $\theta = 0.1$.
Samples of 20 lineages were simulated at each of five time points from both the Kingman ($\lambda_3 = 0$) and Bolthausen-Sznitman ($\Lambda = U(0, 1)$, $\lambda_3 = 0.5$) coalescents.
The data sets are summarised in Table \ref{data}.
Both data sets come from independent simulations, and are sampled from a population at stationarity.
The Kingman coalescent is a classical model of genetic ancestry, while the Bolthausen-Sznitman coalescent has recently been suggested as a ancestral model for influenza and HIV \citep{Neher13}.

\begin{table}[!ht]
\centering
\begin{tabular}{c|r|r}
Time & \multicolumn{1}{c|}{Bolthausen-Sznitman} & \multicolumn{1}{c}{Kingman} \\ \hline
0.0 & 20 x 1001001111 & 20 x 0010000000 \\ \hline
0.5 & 19 x 1001001111 & 15 x 0010000000 \\
& 1 x 1101001111 &  5 x 0000000000 \\ \hline
1.0 & 20 x 1001001111 & 8 x 0010000000 \\
& & 6 x 0000000000 \\
& & 6 x 0010001000 \\ \hline
1.5 & 19 x 1001001111 & 10 x 0010000000 \\
& 1 x 1001101111 & 6 x 0000000000 \\
& & 4 x 0010001000 \\ \hline
2.0 & 19 x 1001001111 & 16 x 0010000000 \\
& 1 x 1001001110 & 4 x 0010001000
\end{tabular}
\caption{Simulated observed sequences from the two models with $\theta = 0.1$ and $M$ as described in the paragraph above.}
\label{data}
\end{table}

We set $\eta = 10^{-6}$ and specify the prior for the density of $\Lambda$ on $[\eta, 1]$ as a Dirichlet process mixture model of truncated Gaussian kernels.
The base measure is the uniform measure on $[\eta, 1]$, with total mass scaled to 0.1.
Finally, the prior for $\tau^{ -1/2 }$, the standard deviations of the truncated Gaussian kernels was chosen to be the $\operatorname{Beta}(1.0, 3.0)$ distribution on $[0, 1]$.
Truncating the maximal standard deviation at 1 excludes some very flat densities from the support of the prior, but the standard normal density is already very flat across $[\eta, 1]$ and the truncation was found to yield substantial gains in speed of convergence of algorithms.
Note that neither data generating model lies in the support of this prior, but both can be well approximated by members of the support.
The choice of hyperparameters was made because it yields a relatively flat marginal prior for $\lambda_3$, the quantity of interest (c.f.~Figure \ref{long}), and the prior satisfies the requirements of the consistency result in Theorem \ref{ts_consistency}.

We make use of the Sethumaran stick-breaking construction of the Dirichlet process \citep{Sethuraman94} and truncate our prior after the first four atoms.
For our choice of base measure and concentration parameter this results in a total variation truncation error of order $400 e^{ -30 } \approx 3.7 \times 10^{-11}$ \citep[Theorem 2]{Ishwaran01}.
Any truncation error could be avoided by pushing forward the prior directly onto the space of moment sequences as illustrated in Section \ref{priors}.
The cost is a more computationally expensive prior to sample and evaluate, as well as a higher dimensional parameter space consisting of 98 moments for our data sets. 
We do not investigate this strategy further in this paper.

The four atom truncation results in 11 parameters: four locations and standard deviations of truncated Gaussian kernels, and three stick break points. 
The fourth break point is set to fulfil the constraint of the weights summing to 1.
We propose updates to these parameters using a truncated Gaussian random walk on $[\eta, 1]^4 \times [0, 1]^4 \times [0, 1]^3$ with covariance matrix $0.0025 \operatorname{Id}$.
This scaling was found to result in a reasonable balance of acceptance probability and jump size for the first moment $\lambda_3$.

We approximate the likelihoods required for computing the acceptance probability $a$ using a straightforward adaptation of the optimised importance sampling method of \citet{Koskela15} to the time series setting of \citep{Beaumont03}, but need to specify the number of particles to use.
More particles will result in more accurate approximations, but at greater computational cost.
In \citep{Doucet15} the authors show that tuning the variance of the log likelihood estimator to 1.44 results in efficient algorithms under a wide range of assumptions.
Preliminary simulations showed this was achieved in our setting by choosing 75 particles for Bolthausen-Sznitman data, and 180 particles for Kingman data.
\vskip 11pt
\begin{rmk}
In the context of real data, when the true data generating parameters are not known, optimising the number of particles using trial runs may require an infeasible amount of computation.
In practice, adaptive algorithms, which optimise parameters online, can be used to circumvent this problem.
Recent work by \citet{Sherlock15} has shown how to implement such adaptivity to all four algorithms discussed below, while maintaining ergodicity and the correct stationary distribution of the algorithm.
\end{rmk}

It is well known that the standard, exact pseudo-marginal algorithm suffers from ``sticking" behaviour, where an unusually high likelihood estimator prevents the algorithm from moving for a macroscopic number of steps \citep{Andrieu09}.
The usual solution is to use a noisy version of the algorithm, in which the likelihood estimator is recomputed at each stage. 
This doubles the number of required likelihood evaluations and biases the algorithm into an incorrect stationary distribution, but can greatly reduce the variance of estimates.
We compare both the exact and noisy versions of the pseudo-marginal algorithm in Figure \ref{traces}.
We also investigate the effect of delayed acceptance acceleration \citep{Christen05}, in which proposed moves are first subjected to an accept-reject decision based on an approximate likelihood function that is cheap to compute.
Only samples which are accepted at this first stage are subjected to an accept-reject decision based on the full likelihood estimates, or more specifically a slight modification to ensure that the delayed acceptance mechanism does not affect the stationary distribution of the algorithm.
In the $\Lambda$-coalescent setting approximate likelihoods are readily available in the form of Product of Approximate Conditionals (or PAC) methods \citep[Section 4.3]{Koskela15}, which we use to implement delayed acceptance chains.

Figure \ref{traces} shows trace plots of 20 000 steps from the four algorithms introduced above.
The exact pseudo-marginal algorithm exhibits sticking behaviour as might be expected, but it is surprising to see that the noisy algorithm does not completely eliminate it.
We conjecture that the remaining stickiness in the noisy trace plot is due to multiple, narrow modes in the 11 dimensional posterior.
It is also clear that the bias in the noisy algorithm is confounding the signal in the data, as the traces are much more intermixed than those of the exact algorithm.

Both the noisy and exact pseudo-marginal algorithm are very computationally expensive to run, particularly for the Kingman data set due to the larger number of particles used to estimate likelihoods.
Delayed acceptance acceleration reduces these run times as expected, particularly for the Kingman case.
Both delayed acceptance algorithms also suffer from sticking, and show less clear separation of the traces than the exact algorithm.
They also look very similar to each other.

\begin{figure}[!ht]
\centering
\includegraphics[width = 0.49 \linewidth]{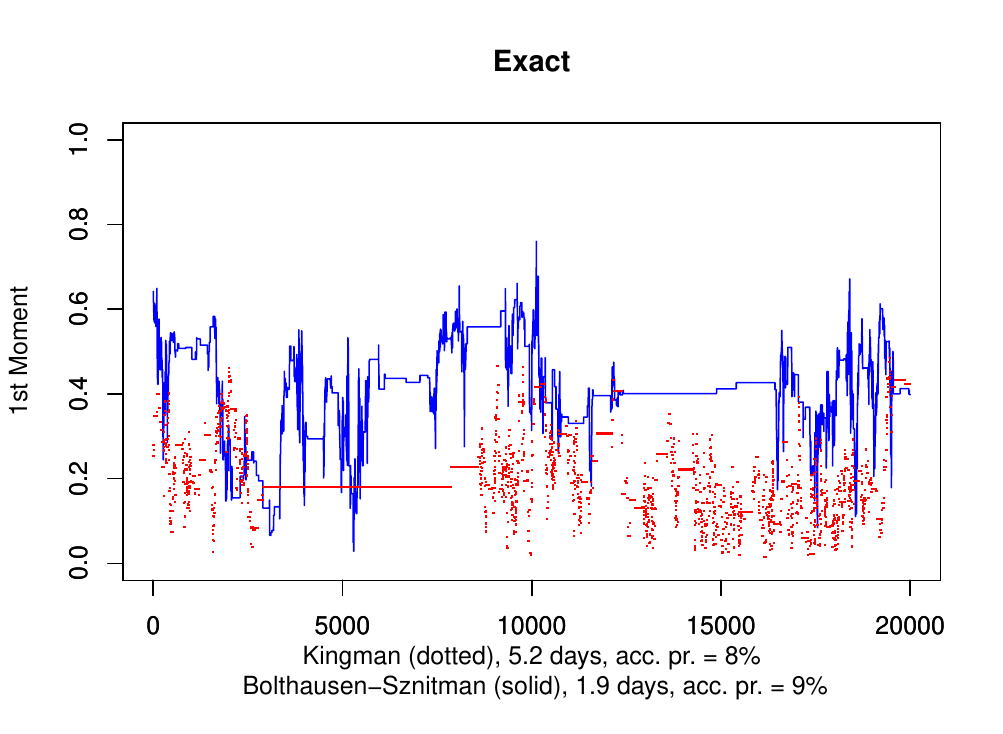}
\includegraphics[width = 0.49 \linewidth]{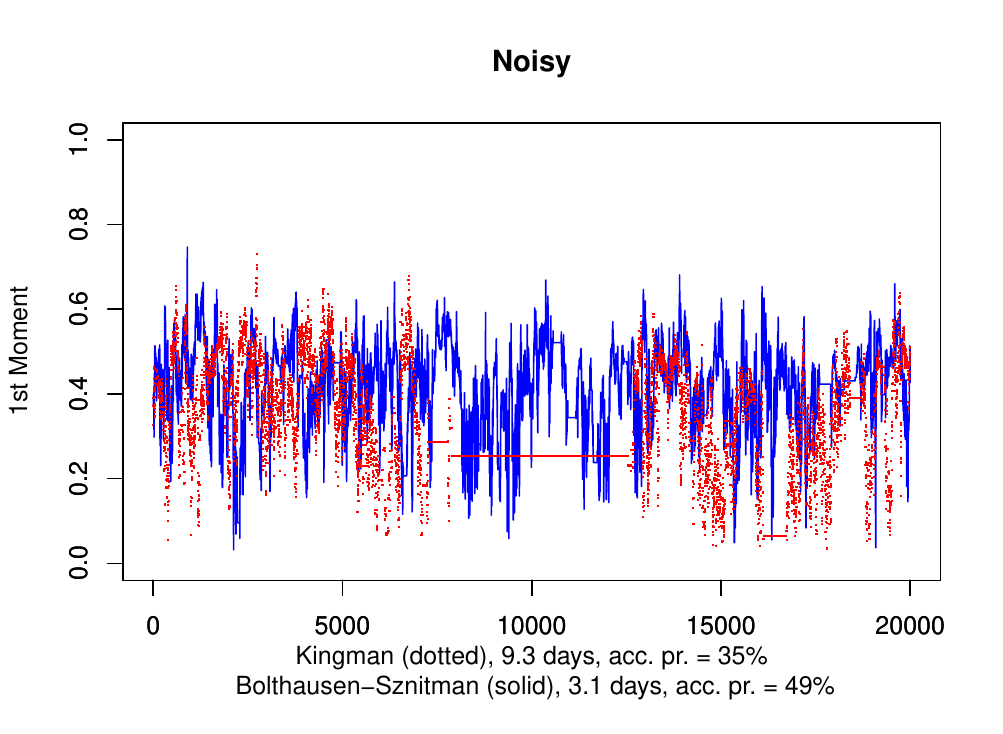}
\includegraphics[width = 0.49 \linewidth]{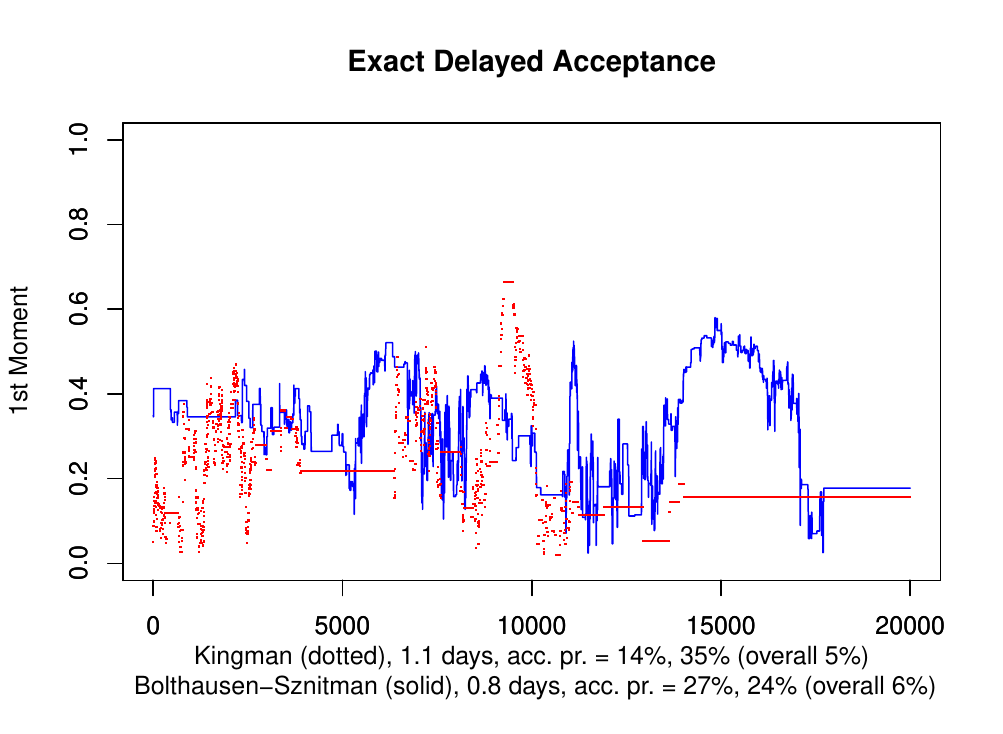}
\includegraphics[width = 0.49 \linewidth]{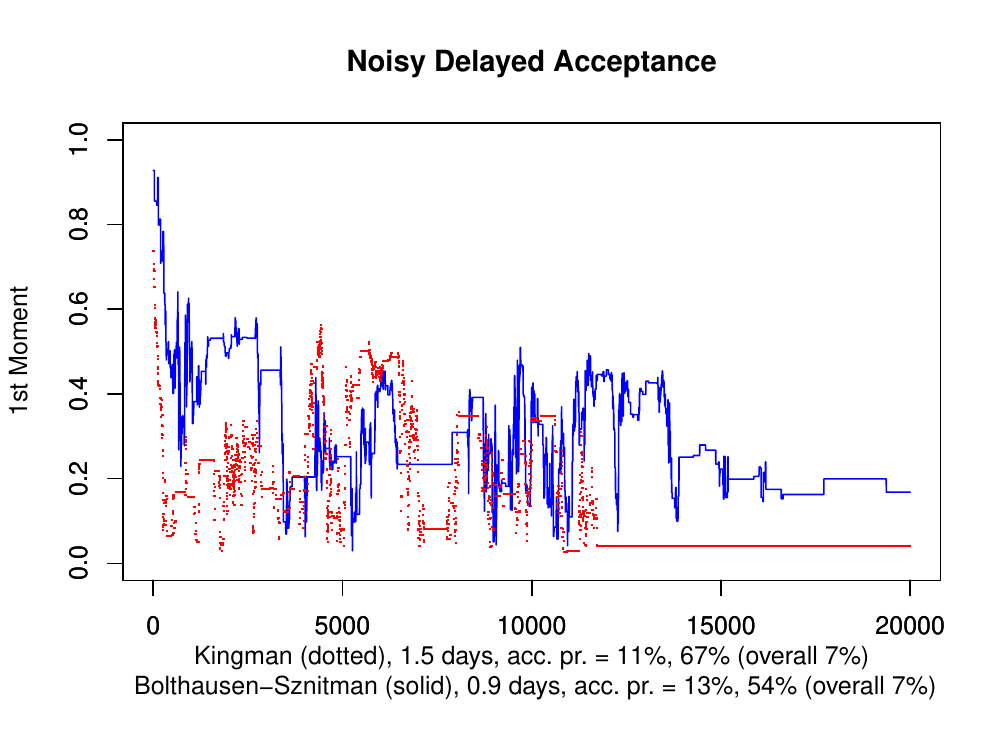}
\caption{Trace plot of the pseudo-marginal algorithm (top left), the noisy algorithm (top right) and corresponding delayed acceptance algorithms (bottom row). Also shown are computation times (on a mid-range Toshiba laptop with an Intel i5 processor) and acceptance probabilities. Delayed acceptance runs show acceptance probabilities for both stages, as well as an overall probability. All runs are independent and initialised from the prior.}
\label{traces}
\end{figure}

Since it appears to be difficult to eliminate sticking behaviour in this case, we chose to leverage the speed up obtained by making use of delayed acceptance and ran a further exact, delayed acceptance pseudo-marginal algorithm for 200 000 steps. 
A trace plot is shown in Figure \ref{long}.
Sticking behaviour is still present, but on a much shorter scale relative to the run length.
Run times are comparable to the noisy algorithm without delayed acceptance, and the Bolthausen-Sznitman trace is again clearly centred at a higher level than the Kingman trace.
We thinned the output of this long run by a factor of 4 000 to reduce the effect of sticking points to obtain 50 samples of first moments, which were used to plot the histograms in Figure \ref{long}.

\begin{figure}[!ht]
\centering
\includegraphics[width = 0.49 \linewidth]{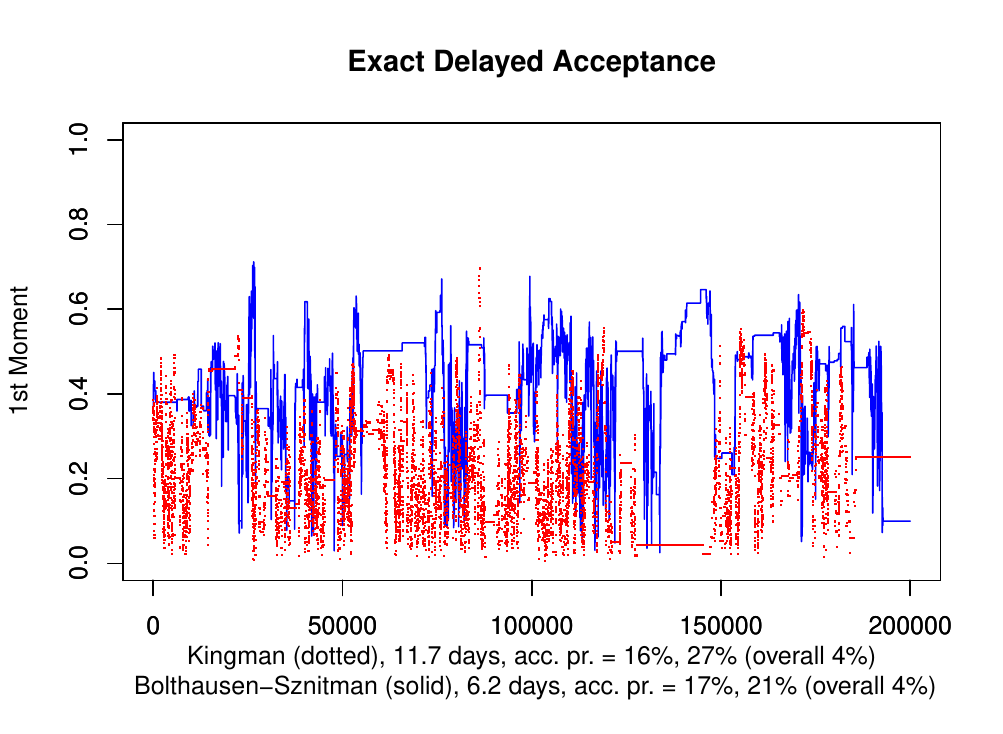}
\includegraphics[width = 0.49 \linewidth]{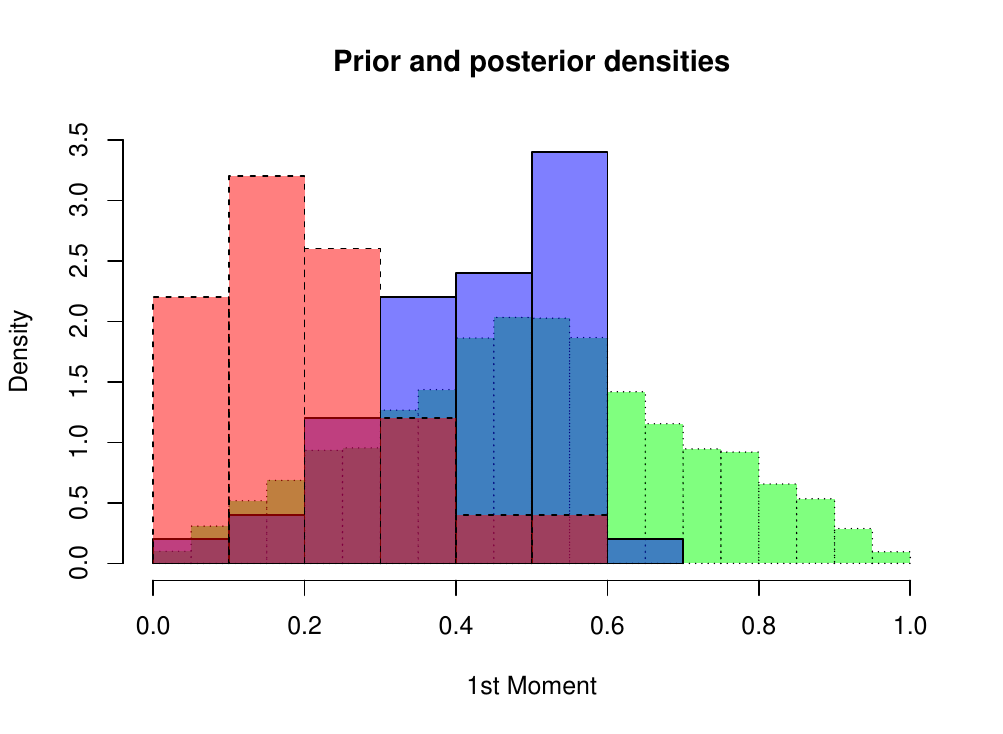}
\caption{(Left) Trace plot of the delayed acceptance exact pseudo-marginal algorithm, along with computation times (on a mid-range Toshiba laptop with an Intel i5 processor) and acceptance probabilities for both stages as well as overall. Both runs are independent and initialised from the prior. (Right) Histograms of both Kingman (dashed) and Bolthausen-Sznitman (solid) posteriors estimated from 50 MCMC samples obtained by thinning the runs shown on the left. The estimated prior density is shown (dotted), based on 10 000 independent samples from the prior.}
\label{long}
\end{figure}

It is clear from both the trace plots and histograms in Figure \ref{long} that the run length is still not sufficient for fully converged estimates.
However, both plots already  show a clear shift of posterior modes toward the values generating the data.
The red histogram is  consistent with the Kingman coalescent, while the blue one is consistent with the Bolthausen-Sznitman coalescent.
Moreover, approximate 95\% credible intervals are $\lambda_3 \in [0.1, 0.6]$ for the Bolthausen-Sznitman posterior, and $\lambda_3 \in [\eta, 0.5]$ for the Kingman posterior.
This suggests the relatively short time series is nevertheless sufficiently informative to reject the incorrect model in both cases.

\section{Discussion}\label{discussion}

In this paper we have presented a robust framework for Bayesian non-parametric inference under $\Lambda$-coalescent processes for time series data, and studied the feasibility of implementable families of algorithms for practical inference.
We demonstrated that time series data is necessary for consistent inference, and gave verifiable conditions for posterior consistenct.
As seen in Example \ref{ex_inconsistency}, lack of consistency can lead to very low statistical power and high sensitivity of inference both to confounding parameters, such as mutation rate, and the observed allele frequencies.
A theoretical guarantee of consistency is crucial as expressions for statistical power rely on intractable stationary distributions and transition densities of $\Lambda$-Fleming-Viot jump-diffusions, making the reliability of experiments without time series data very difficult to evaluate.

Efficient methods for importance sampling $\Lambda$-coalescent trees are available \citep{Birkner11, Koskela15}, and these can be used to generalise the pseudo-marginal MCMC algorithms of \citep{Beaumont03} for temporally spaced data.
The consistency conditions of Theorem \ref{ts_consistency} on the prior are sufficiently mild to permit the use of Dirichlet process mixture model priors, which can be readily truncated for implementable algorithms.
Alternatively, we have shown that parametrising the inference problem via truncated moment sequences leads to implementable algorithms with no discretisation or truncation error.
This work provides a strong indication that time series data, and accompanying inference methods such as the one outlined above, should be adopted as standard whenever the coalescent generating the data cannot be assumed to be known.

Generalising of consistency result within the $\Lambda$-Fleming-Viot process class to include unknown drift, which can be used to model e.g.~mutation, recombination and selection, as well as more general $\Lambda$-measures is of great interest.
However, it is difficult for a number of reasons.
Firstly, relaxing conditions on $\Lambda$ near 0 while ensuring the integral in \eqref{kl_cond} remains finite is challenging.
Likewise, it is well known that equivalent changes of measure for L\'evy processes necessitate equivalent L\'evy measures (see e.g.~\citep{Sato99}, Theorem 33.1), and this is also the condition needed for the jump-diffusions considered in \citep{Cheridito05}.
The way in which the drift can be transformed while maintaining absolute continuity in \citep{Cheridito05} is also restrictive, and depends on the diffusion coefficient and L\'evy compensator.
Finally, any difference in diffusion coefficients will obviously destroy absolute continuity outright, so if there were an atom $\Lambda( \{ 0 \} ) > 0$, its size would have to be known with certainty.

It would also be of great interest to obtain contraction rates of the posterior under verifiable conditions. 
Obtaining rates is a challenging problem in non-i.i.d.~Bayesian non-parametric inference, and existing results by \citet{Gugushvili15} for compound Poisson processes and \citet{Nickl15} for scalar diffusions do not seem generalisable.
A different approach by \citet{Nguyen13} for mixing measures of infinite mixture models could present a promising directions of future work by viewing the $\Lambda$-coalescent tree as a mixture of merger events, but adaptation into the present setting is a formidable task and is beyond the scope of this paper.

The method of parametrising the unknown $\Lambda$-measure by its first $n - 2$ moments when the data set is of size $n \in \N$ reflects the limited amount of signal in finite data.
More precisely, the likelihood given a sample of size $n \in \N$ is constant within moment classes of order $n$, so that any variation in the posterior within these moment classes is due solely to the prior.
Hence this parametrisation can be seen as regularising an under-determined inference problem in an infinite dimensional space by identifying an appropriate, data-driven, finite dimensional quotient space in which to conduct inference.
We believe this approach to have more broad applicability in non-parametric statistics as well as an alternative to direct regularisation by a prior in the infinite dimensional space, or to approximate projections onto finite dimensional subspaces \citep{Cui14}.

The algorithms used to approximate the posterior and maximise/minimise quantities of interest given the posterior are highly computationally intensive, and we do not expect our approach to be competitive with well-chosen parametric families when the number of observed lineages or loci is large.
However, the simulations in Section \ref{simulation} demonstrate that our assumption-free framework can be used to empirically evaluate the modelling fit of parametric families given moderately sized pilot data, for instance by ensuring that the family contains a candidate $\Lambda$ which matches the MAP estimators of some small number of moments.
Such parametric families can then be confidently used to process larger data sets.
The pseudo-marginal method can also be adapted to incorporate unknown mutation parameters, recombination and other forces not considered in this paper, albeit at the cost of greater computational cost and lower parameter identifiability.
This cost can be alleviated to a large extent by modern GPU and cluster computing approaches, because the importance sampling algorithm used to estimate likelihoods is readily parallelisable.
For example, up to 500 fold speed up was reported by \citet{Lee10} when computations were parallelised on GPUs instead of being run in serial on CPUs.
Such gains in computation speed would make the algorithms employed in Section \ref{simulation} practical for many realistic genetic data sets.

\section*{Acknowledgements}

The authors are grateful to Tim Sullivan for insight into orthogonal polynomials as well as Bayesian brittleness, to Felipe Medina Aguayo for fruitful conversations about pseudo-marginal methods, to Yun Song for helpful comments on Remark \ref{yuns_result}, and to Matthias Birkner for assistance with identifiability conditions.
Jere Koskela was supported by EPSRC as part of the MASDOC DTC at the University of Warwick. Grant No. EP/HO23364/1. 
Paul Jenkins is supported in part by EPSRC grant EP/L018497/1.

\bibliography{science}  

\begin{thebibliography}{58}
\providecommand{\natexlab}[1]{#1}
\providecommand{\url}[1]{\texttt{#1}}
\expandafter\ifx\csname urlstyle\endcsname\relax
  \providecommand{\doi}[1]{doi: #1}\else
  \providecommand{\doi}{doi: \begingroup \urlstyle{rm}\Url}\fi

\bibitem[Anderson(2005)]{Anderson05}
E.~C. Anderson.
\newblock An efficient {Monte Carlo} method for estimating {$N_e$} from
  temporally spaced samples using a coalescent-based likelihood.
\newblock \emph{Genetics}, 170\penalty0 (2):\penalty0 955--967, 2005.

\bibitem[Andrieu and Roberts(2009)]{Andrieu09}
C.~Andrieu and G.~O. Roberts.
\newblock The pseudo-marginal approach for efficient {Monte Carlo}
  computations.
\newblock \emph{Ann. Stat.}, 37\penalty0 (2):\penalty0 697--725, 2009.

\bibitem[\'{A}rnason(2004)]{Arnason04}
E.~\'{A}rnason.
\newblock Mitochondrial cytochrome b {DNA} variation in the high-fecundity
  {A}tlantic cod: trans-{A}tlantic clines and shallow gene genealogy.
\newblock \emph{Genetics}, 166:\penalty0 1871--1885, 2004.

\bibitem[Beaumont(2003)]{Beaumont03}
M.~A. Beaumont.
\newblock Estimation of population growth or decline in genetically monitored
  populations.
\newblock \emph{Genetics}, 164:\penalty0 1139--1160, 2003.

\bibitem[Berestycki et~al.(2013)Berestycki, Berestycki, and
  Schweinsberg]{Berestycki13b}
J.~Berestycki, N.~Berestycki, and J.~Schweinsberg.
\newblock The genealogy of branching {Brownian} motion with absorption.
\newblock \emph{Ann. Probab.}, 41\penalty0 (2):\penalty0 527--618, 2013.

\bibitem[Bertoin and Le~Gall(2003)]{Bertoin03}
J.~Bertoin and J.-F. Le~Gall.
\newblock Stochastic flows associated to coalescent processes.
\newblock \emph{Probab. Theory Related Fields}, 126:\penalty0 261--288, 2003.

\bibitem[Bhattacharya and Dunson(2012)]{Bhattacharya12}
A.~Bhattacharya and D.~B. Dunson.
\newblock Strong consistency of nonparametric {Bayes} density estimation on
  compact metric spaces with applications to specific manifolds.
\newblock \emph{Ann. Inst. Stat. Math.}, 64:\penalty0 687--714, 2012.

\bibitem[Birkner and Blath(2008)]{Birkner08}
M.~Birkner and J.~Blath.
\newblock Computing likelihoods for coalescents with multiple collisions in the
  infinitely many sites model.
\newblock \emph{J. Math. Biol.}, 57\penalty0 (3):\penalty0 435--463, 2008.

\bibitem[Birkner and Blath(2009)]{Birkner09b}
M.~Birkner and J.~Blath.
\newblock Measure-valued diffusions, general coalescents and population genetic
  inference.
\newblock \emph{in J. Blath, P. M\"{o}rters, M. Scheutzow (Eds.), Trends in
  Stochastic Analysis}, LMS 351:\penalty0 329--363, 2009.

\bibitem[Birkner et~al.(2011)Birkner, Blath, and Steinr\"{u}cken]{Birkner11}
M.~Birkner, J.~Blath, and M.~Steinr\"{u}cken.
\newblock Importance sampling for {L}ambda--coalescents in the infinitely many
  sites model.
\newblock \emph{Theor. Popln Biol.}, 79\penalty0 (4):\penalty0 155--173, 2011.

\bibitem[Bollback et~al.(2008)Bollback, York, and Nielsen]{Bollback08}
J.~P. Bollback, T.~L. York, and R.~Nielsen.
\newblock Estimation of {$2 N_e s$} from temporal allele frequency data.
\newblock \emph{Genetics}, 179\penalty0 (1):\penalty0 497--502, 2008.

\bibitem[Boom et~al.(1994)Boom, Boulding, and Beckenback]{Boom94}
J.~D.~G. Boom, E.~G. Boulding, and A.~T. Beckenback.
\newblock Mitochondrial {DNA} variation in introduced populations of {P}acific
  oyster, {C}rassostrea gigas, in {B}ritish {C}olumbia.
\newblock \emph{Can. J. Fish. Aquat. Sci.}, 51:\penalty0 1608--1614, 1994.

\bibitem[Cheridito et~al.(2005)Cheridito, Filipovi\'c, and Yor]{Cheridito05}
P.~Cheridito, D.~Filipovi\'c, and M.~Yor.
\newblock Equivalent and absolutely continuous measure changes for
  jump-diffusion processes.
\newblock \emph{Ann. Appl. Probab.}, 15\penalty0 (3):\penalty0 1713--1732,
  2005.

\bibitem[Christen and Fox(2005)]{Christen05}
J.~A. Christen and C.~Fox.
\newblock Markov chain {Monte Carlo} using an approximation.
\newblock \emph{J. Comput. Graph. Stat.}, 14\penalty0 (4):\penalty0 795--810,
  2005.

\bibitem[Cui et~al.(2014)Cui, Martin, Marzouk, Solonen, and Spantini]{Cui14}
T.~Cui, J.~Martin, Y.~M. Marzouk, A.~Solonen, and A.~Spantini.
\newblock Likelihood-informed dimension reduction for nonlinear inverse
  problems.
\newblock \emph{Inverse Problems}, 30\penalty0 (11):\penalty0 114015, 2014.

\bibitem[Der and Plotkin(2014)]{Der14}
R.~Der and J.~B. Plotkin.
\newblock The equilibrium allele frequency distribution for a population with
  reproductive skew.
\newblock \emph{Genetics}, 196\penalty0 (4):\penalty0 1199--1216, 2014.

\bibitem[Diaconis and Freedman(1986)]{Diaconis86}
P.~Diaconis and D.~Freedman.
\newblock On the consistency of {Bayes} estimates.
\newblock \emph{Ann. Statist.}, 14\penalty0 (1):\penalty0 1--26, 1986.

\bibitem[Donnelly and Kurtz(1999)]{Donnelly99}
P.~Donnelly and T.~Kurtz.
\newblock Particle representations for measure-valued population models.
\newblock \emph{Ann. Probab.}, 27\penalty0 (1):\penalty0 166--205, 1999.

\bibitem[Doucet et~al.(2015)Doucet, Pitt, Deligiannidis, and Kohn]{Doucet15}
A.~Doucet, M.~Pitt, G.~Deligiannidis, and R.~Kohn.
\newblock Efficient implementation of {Markov} chain {Monte Carlo} when using
  an unbiased likelihood estimator.
\newblock \emph{Biometrika}, 102\penalty0 (2):\penalty0 295--313, 2015.

\bibitem[Drummond et~al.(2002)Drummond, Nicholls, Rodrigo, and
  Solomon]{Drummond02}
A.~J. Drummond, G.~K. Nicholls, A.~G. Rodrigo, and W.~Solomon.
\newblock Estimating mutation parameters, population history and genealogy
  simultaneously from temporally spaced sequence data.
\newblock \emph{Genetics}, 161\penalty0 (3):\penalty0 1307--1320, 2002.

\bibitem[Drummond et~al.(2005)Drummond, Rambaut, Shapiro, and
  Pybus]{Drummond05}
A.~J. Drummond, A.~Rambaut, B.~Shapiro, and O.~G. Pybus.
\newblock Bayesian coalescent inference of past population dynamics from
  molecular sequences.
\newblock \emph{Mol. Biol. Evol.}, 22:\penalty0 1185--1192, 2005.

\bibitem[Durrett and Schweinsberg(2005)]{Durrett05}
R.~Durrett and J.~Schweinsberg.
\newblock A coalescent model for the effect of advantageous mutations on the
  genealogy of a population.
\newblock \emph{Stoch. Proc. Appl.}, 115:\penalty0 1628--1657, 2005.

\bibitem[Eldon and Wakeley(2006)]{Eldon06}
B.~Eldon and J.~Wakeley.
\newblock Coalescent processes when the distribution of offspring number among
  individuals is highly skewed.
\newblock \emph{Genetics}, 172:\penalty0 2621--2633, 2006.

\bibitem[Ferguson(1973)]{Ferguson73}
T.~S. Ferguson.
\newblock A {Bayesian} analysis of some nonparametric problems.
\newblock \emph{Ann. Stat.}, 1\penalty0 (2):\penalty0 209--230, 1973.

\bibitem[Griffiths(2014)]{Griffiths14}
R.~C. Griffiths.
\newblock The {$\Lambda$}-{Fleming}-{Viot} process and a connection with
  {Wright}-{Fisher} diffusion.
\newblock \emph{Adv. Appl. Probab.}, 46\penalty0 (4):\penalty0 1009--1035,
  2014.

\bibitem[Gugushvili et~al.(2015)Gugushvili, van~der Meulen, and
  Spreij]{Gugushvili15}
S.~Gugushvili, F.~van~der Meulen, and P.~Spreij.
\newblock Nonparametric {Bayesian} inference for multidimensional compound
  {Poisson} processes.
\newblock \emph{Mod. Stoch. Theory Appl.}, 2\penalty0 (1):\penalty0 1--15,
  2015.

\bibitem[Gurland(1948)]{Gurland48}
J.~Gurland.
\newblock Inversion formulae for the distribution of ratios.
\newblock \emph{Ann. Math. Statist.}, 19:\penalty0 228--237, 1948.

\bibitem[Hannum et~al.(1981)Hannum, Hollander, and Langberg]{Hannum81}
R.~C. Hannum, M.~Hollander, and N.~A. Langberg.
\newblock Distributional results for random functionals of a {Dirichlet}
  process.
\newblock \emph{Ann. Probab.}, 9:\penalty0 665--670, 1981.

\bibitem[Hjort et~al.(2010)Hjort, Holmes, M\"uller, and Walker]{Hjort10}
N.~L. Hjort, C.~Holmes, P.~M\"uller, and S.~G. Walker, editors.
\newblock \emph{Bayesian nonparametrics}.
\newblock Cambridge series in statistical and probabilistic mathematics.
  Cambridge University Press, 2010.

\bibitem[Hoffmann and Olivier(2016)]{Hoffmann16}
M.~Hoffmann and A.~Olivier.
\newblock Nonparametric estimation of the division rate of an age dependent
  branching process.
\newblock \emph{Stoch. Proc. Appl.}, 126\penalty0 (5):\penalty0 1433--1471,
  2016.

\bibitem[H\"opfner et~al.(2002)H\"opfner, Hoffmann, and
  L\"ocherbach]{Hopfner02}
R.~H\"opfner, M.~Hoffmann, and E.~L\"ocherbach.
\newblock Nonparametric estimation of the death rate in branching diffusions.
\newblock \emph{Scand. J. Stat.}, 29:\penalty0 665--690, 2002.

\bibitem[Ishwaran and James(2001)]{Ishwaran01}
H.~Ishwaran and L.~F. James.
\newblock Gibbs sampling methods for stick-breaking priors.
\newblock \emph{J. Amer. Statist. Assoc.}, 96\penalty0 (453):\penalty0
  161--173, 2001.

\bibitem[Kingman(1982)]{Kingman82}
J.~F.~C. Kingman.
\newblock The coalescent.
\newblock \emph{Stochast. Process. Applic.}, 13\penalty0 (3):\penalty0
  235--248, 1982.

\bibitem[Koskela et~al.(2015)Koskela, Jenkins, and Span\`{o}]{Koskela15}
J.~Koskela, P.~A. Jenkins, and D.~Span\`{o}.
\newblock Computational inference beyond {Kingman's} coalescent.
\newblock \emph{J. Appl. Probab.}, 52\penalty0 (2):\penalty0 519--537, 2015.

\bibitem[Koskela et~al.(2017)Koskela, Span\`o, and Jenkins]{Koskela17}
J.~Koskela, D.~Span\`o, and P.~A. Jenkins.
\newblock Consistency of {Bayesian} nonparametric inference for discretely
  observed jump diffusions.
\newblock \emph{Preprint}, arXiv:1506.04709, 2017.

\bibitem[Lee et~al.(2010)Lee, Yau, Giles, Doucet, and Holmes]{Lee10}
A.~Lee, C.~Yau, M.~B. Giles, A.~Doucet, and C.~C. Holmes.
\newblock On the utility of graphics cards to perform massively parallel
  simulation of advanced {Monte Carlo} methods.
\newblock \emph{J. Comp. Graph. Stat.}, 19\penalty0 (4):\penalty0 769--789,
  2010.

\bibitem[Lo(1984)]{Lo84}
A.~Y. Lo.
\newblock On a class of {Bayesian} nonparametric estimates. 1. density
  estimates.
\newblock \emph{Ann. Statist.}, 12:\penalty0 351--357, 1984.

\bibitem[Malaspinas et~al.(2012)Malaspinas, Malaspinas, Evans, and
  Slatkin]{Malaspinas12}
A.-S. Malaspinas, O.~Malaspinas, S.~N. Evans, and M.~Slatkin.
\newblock Estimating allele age and selection coefficient from time-serial
  data.
\newblock \emph{Genetics}, 192\penalty0 (2):\penalty0 599--607, 2012.

\bibitem[Mathieson and McVean(2013)]{Mathieson13}
I.~Mathieson and G.~McVean.
\newblock Estimating selection coefficients in spatially structured populations
  from time series data of allele frequencies.
\newblock \emph{Genetics}, 193\penalty0 (3):\penalty0 973--984, 2013.

\bibitem[Minin et~al.(2008)Minin, Bloomquist, and Suchard]{Minin08}
V.~N. Minin, E.~W. Bloomquist, and M.~A. Suchard.
\newblock Smooth skyride through a rough skyline: {Bayesian} coalescent-based
  inference of population dynamics.
\newblock \emph{Mol. Biol. Evol.}, 25:\penalty0 1459--1471, 2008.

\bibitem[M\"ohle(1999)]{Mohle99}
M.~M\"ohle.
\newblock The concept of duality and applications to {Markov} processes arising
  in neutral population genetics models.
\newblock \emph{Bernoulli}, 5\penalty0 (5):\penalty0 761--777, 1999.

\bibitem[M\"ohle(2006)]{Mohle06}
M.~M\"ohle.
\newblock On sampling distributions for coalescent processes with simultaneous
  multiple collisions.
\newblock \emph{Bernoulli}, 12\penalty0 (1):\penalty0 35--53, 2006.

\bibitem[Neher and Hallatschek(2013)]{Neher13}
R.~A. Neher and O.~Hallatschek.
\newblock Genealogies of rapidly adapting populations.
\newblock \emph{Proc. Natl Acad. Sci.}, 110\penalty0 (2):\penalty0 437--442,
  2013.

\bibitem[Nguyen(2013)]{Nguyen13}
X.~Nguyen.
\newblock Convergence of latent mixing measures in finite and infinite mixture
  models.
\newblock \emph{Ann. Stat.}, 41\penalty0 (1):\penalty0 370--400, 2013.

\bibitem[Nickl and S\"ohl(2015)]{Nickl15}
R.~Nickl and J.~S\"ohl.
\newblock Nonparametric {Bayesian} posterior contraction rates for discretely
  observed scalar diffusions.
\newblock \emph{Preprint}, arXiv:1510.05526, 2015.

\bibitem[Owhadi et~al.(2015)Owhadi, Scovel, and Sullivan]{Owhadi15}
H.~Owhadi, C.~Scovel, and T.~Sullivan.
\newblock Brittleness of {Bayesian} inference under finite information in a
  continuous world.
\newblock \emph{Electron. J. Statist.}, 9\penalty0 (1):\penalty0 1--79, 2015.

\bibitem[Pitman(1999)]{Pitman99}
J.~Pitman.
\newblock Coalescents with multiple collisions.
\newblock \emph{Ann. Probab.}, 27\penalty0 (4):\penalty0 1870--1902, 1999.

\bibitem[Regazzini et~al.(2002)Regazzini, Guglielmi, and Di~Nunno]{Regazzini02}
E.~Regazzini, A.~Guglielmi, and G.~Di~Nunno.
\newblock Theory and numerical analysis for exact distributions of functionals
  of a {Dirichlet} process.
\newblock \emph{Ann. Stat.}, 30\penalty0 (5):\penalty0 1376--1411, 2002.

\bibitem[Sagitov(1999)]{Sagitov99}
S.~Sagitov.
\newblock The general coalescent with asynchronous mergers of ancestral
  lineages.
\newblock \emph{J. Appl. Probab.}, 36\penalty0 (4):\penalty0 1116--1125, 1999.

\bibitem[Sato(1999)]{Sato99}
K.-I. Sato.
\newblock \emph{L\'evy processes and infinitely divisible distributions}.
\newblock Cambridge University Press, 1999.

\bibitem[Schweinsberg(2015)]{Schweinsberg15}
J.~Schweinsberg.
\newblock Rigorous results for a population model with selection {II}:
  genealogy of the population.
\newblock \emph{Preprint}, arXiv:1507.00394, 2015.

\bibitem[Sethuraman(1994)]{Sethuraman94}
J.~Sethuraman.
\newblock A constructive definition of {Dirichlet} priors.
\newblock \emph{Stat. Sinica}, 4:\penalty0 639--650, 1994.

\bibitem[Sherlock et~al.(2015)Sherlock, Golightly, and Henderson]{Sherlock15}
C.~Sherlock, A.~Golightly, and D.~A. Henderson.
\newblock Adaptive, delayed-acceptance {MCMC} for targets with expensive
  likelihoods.
\newblock \emph{Preprint}, arXiv:1509.00172, 2015.

\bibitem[Spence et~al.(2016)Spence, Kamm, and Song]{Spence16}
J.~P. Spence, J.~A. Kamm, and Y.~S. Song.
\newblock The site frequency spectrum for general coalescents.
\newblock \emph{Genetics}, 202:\penalty0 1549--1561, 2016.

\bibitem[Steinr\"{u}cken et~al.(2013)Steinr\"{u}cken, Birkner, and
  Blath]{Steinrucken13a}
M.~Steinr\"{u}cken, M.~Birkner, and J.~Blath.
\newblock Analysis of {DNA} sequence variation within marine species using
  {B}eta--coalescents.
\newblock \emph{Theor. Popln Biol.}, 87:\penalty0 15--24, 2013.

\bibitem[van~der Meulen and van Zanten(2013)]{vanDerMeulen13}
F.~van~der Meulen and H.~van Zanten.
\newblock Consistent nonparametric {Bayesian} inference for discretely observed
  scalar diffusions.
\newblock \emph{Bernoulli}, 19\penalty0 (1):\penalty0 44--63, 2013.

\bibitem[Wang(2010)]{Wang10}
J.~Wang.
\newblock Regularity of semigroups generated by {L}\'evy type operators via
  coupling.
\newblock \emph{Stoch. Proc. Appl.}, 120\penalty0 (9):\penalty0 1680--1700,
  2010.

\bibitem[Winkler(1988)]{Winkler88}
G.~Winkler.
\newblock Extreme points of moment sets.
\newblock \emph{Math. Oper. Res.}, 30\penalty0 (4):\penalty0 581--587, 1988.

\end{thebibliography}
\bibliographystyle{plainnat}

\end{document}